\documentclass[11pt]{article}

\usepackage{sn-preamble} 
\usepackage{tikz}
\usepackage{verbatim}
\usepackage{cancel}
\usepackage{caption}

\usetikzlibrary{decorations.pathreplacing,angles,quotes}
\usetikzlibrary{shapes.geometric}
\usetikzlibrary{patterns}

\usepackage{mathtools} 
\usepackage{bm} 

\newcommand{\red}[1]{{\color{red} {#1}}}

\newcommand{\ignore}[1]{}

\newcommand{\duc}{\mathsf{dist}_{\mathsf{UC}}}
\newcommand{\dint}{\mathsf{dist}_{\mathsf{int}}}

\newcommand{\view}{\mathrm{view}}
\newcommand{\Dyes}{\calD_{\mathrm{yes}}}
\newcommand{\Dno}{\calD_{\mathrm{no}}}
\newcommand{\fyes}{\boldf_{\mathrm{yes}}}
\newcommand{\fno}{\boldf_{\mathrm{no}}}
\newcommand{\yes}{\mathrm{yes}}
\newcommand{\no}{\mathrm{no}}
\newcommand{\Tal}{\mathsf{Talagrand}}

\begin{document}

\title{Testing Intersecting and Union-Closed Families\vspace{0.5em}}

\author{
Xi Chen \thanks{Columbia University. Email: \url{xichen@cs.columbia.edu}.}\hspace{-0.1cm} \and 
Anindya De \thanks{University of Pennsylvania. Email: \url{anindyad@seas.upenn.edu}.}\hspace{-0.1cm} \and 
Yuhao Li \thanks{Columbia University. Email: \url{yuhaoli@cs.columbia.edu}.}\hspace{-0.1cm} \and 
Shivam Nadimpalli \thanks{Columbia University. Email: \url{sn2855@columbia.edu}.}\hspace{-0.1cm} \and 
Rocco A. Servedio \thanks{Columbia University. Email: \url{rocco@cs.columbia.edu}.} 
\vspace{0.5em}
}

\date{
\small\today
}

\pagenumbering{gobble}

\maketitle  

\begin{abstract}

Inspired by the classic problem of Boolean function monotonicity testing, 
we investigate the testability of other well-studied properties of combinatorial finite set systems, specifically  \emph{intersecting} families and \emph{union-closed} families.  A function $f: \zo^n \to \zo$ is intersecting (respectively, union-closed) if its set of satisfying assignments corresponds to an intersecting family (respectively, a union-closed family) of subsets of $[n]$.  

Our main results are that --- in sharp contrast with the property of being a monotone set system ---  the property of being an intersecting set system, and the property of being a union-closed set system, both turn out to be information-theoretically difficult to test.  We show that: 

\begin{itemize}

\item For $\eps \geq \Omega(1/\sqrt{n})$, any non-adaptive two-sided $\eps$-tester for intersectingness must make $2^{\Omega(n^{1/4}/\sqrt{\eps})}$ queries. We also give a 
$2^{\Omega(\sqrt{n \log(1/\eps)})}$-query
lower bound for non-adaptive one-sided $\eps$-testers for intersectingness.

\item For $\eps \geq 1/2^{\Omega(n^{0.49})}$, any non-adaptive two-sided $\eps$-tester for union-closedness must make $n^{\Omega(\log(1/\eps))}$ queries.

\end{itemize}
Thus, neither intersectingness nor union-closedness shares the $\poly(n,1/\eps)$-query non-adaptive testability that is enjoyed by monotonicity.

To complement our lower bounds, we also give a simple $\poly(n^{\sqrt{n\log(1/\eps)}},1/\eps)$-query, one-sided, non-adaptive algorithm for $\eps$-testing each of these properties (intersectingness and union-closedness).  We thus achieve nearly tight upper and lower bounds for two-sided testing of intersectingness when $\eps = \Theta(1/\sqrt{n})$, and for one-sided testing of intersectingness when $\eps=\Theta(1).$

\end{abstract}


\newpage
\setcounter{page}{1}
\pagenumbering{arabic}


\section{Introduction}
\label{sec:intro}

 Monotonicity testing is among the oldest and most intensively studied problems in property testing (see e.g. \cite{GGLRS,DGLRRS,FLNRRS,HalevyKushilevitz:07,BCGM12,CS13a,CS13b,CST14,CDST15,BB15,CWX17stoc,BCS18,KMS18,CS19,BCS20,PRW22,BCS23,BKKM23}
 and the numerous references contained therein).  The simplicity with which the core monotonicity testing problem can be formulated---given query access to an unknown $f: \zo^n \to \zo$, output ``yes'' if $f$ is monotone and ``no'' if $f$ is far in Hamming distance from every monotone function---belies the wealth of sophisticated technical ingredients and ideas (such as combinatorial shifting \cite{GGLRS,DGLRRS}, multidimensional limit theorems \cite{CST14,CDST15}, and isoperimetric inequalities \cite{CS13a,KMS18,BCS18,PRW22,BCS23,BKKM23}) which have been deployed in both algorithms and lower bounds for this problem.  Thanks to this body of work the basic problem of monotonicity testing is now fairly well understood:  \cite{KMS18} gave an $\tilde{O}(\sqrt{n}/\eps^2)$-query non-adaptive testing algorithm, and \cite{CWX17stoc} gave an $\tilde{\Omega}(n^{1/3})$-query lower bound which holds even for adaptive algorithms.

Monotonicity testing has several intriguing features as a property testing problem: 
\begin{itemize} 

\item Since the class of all monotone functions is of doubly exponential size\footnote{Observe that any assignment of 0/1 values to the middle level of the Boolean hypercube $\zo$ corresponds to at least one monotone function, and hence there are at least $2^{\Omega(2^n/\sqrt{n})}$ many distinct monotone functions over $\zo^n$.}, the results mentioned above tell us that the query complexity of testing this class, which contains $N=2^{2^{\Theta(n)}}$ functions, is $(\log \log N)^c$ for some constant ${\frac 1 3} \leq c \leq {\frac 1 2}$. This is an interesting contrast with both the $O(\log N)$ query complexity which suffices to test any class of $N$ functions\footnote{This follows straightforwardly from the fact that $O(\log N)$ samples suffice to properly PAC learn any concept class of $N$ Boolean functions \cite{bluehrhauwar89} and the well-known reduction from proper PAC learning to property testing given in \cite{GGR98}.}  and the constant query complexity (independent of $N$ and depending only on the error parameter $\eps$) of  a number of other well-studied property testing problems such as linearity testing \cite{BLR93}, testing linear separability \cite{MORS10}, and testing dictatorship \cite{PRS02}.

\item The monotonicity of  $f: \zo^n \to \zo$ is equivalent to having all pairs of inputs $x,y$ satisfy a simple ``pair condition,'' which is that 
\begin{equation} \label{eq:mono}
x \leq y \implies f(x) \leq f(y).
\end{equation}
Given this, it is natural to consider ``pair testers'' for monotonicity which work by drawing a pair of inputs $\bx,\by \in \zo^n$ with $\bx \leq \by$ according to some distribution over such pairs, and checking whether the pair violates monotonicity.  Indeed, all known algorithms for testing monotonicity, including the state-of-the-art algorithm of \cite{KMS18}, work in this fashion.

\item Finally, we observe that a monotone function $f: \zo^n \to \zo$ can alternately be viewed as an \emph{upward-closed} set system: this is a collection of subsets ${\cal S} \subseteq 2^{[n]}$, corresponding to the satisfying assignments of $f$, which has the property that for every subset $S \subseteq [n]$, if $S \in {\cal S}$ then $S \cup \{i\} \in {\cal S}$ for every $i \in [n]$.
\end{itemize}

\noindent {\bf This Work.}  Motivated by monotonicity testing, we propose to study other combinatorial property testing problems of a similar flavor. In particular, we are interested in the testability of properties which (a) are ``very large'' (meaning that the number of functions with the property is doubly exponential in $n$); (b) are defined by a natural condition on pairs or triples of inputs; and (c) correspond to well-studied properties of set systems. We focus on two specific properties of this sort, namely \emph{intersecting} and \emph{union-closed} set systems.

\medskip 

\noindent {\bf Intersectingness.}
A set system ${\cal S} \subseteq 2^{[n]}$ is said to be \emph{intersecting} if any two sets $S_1, S_2 \in {\cal S}$ have a nonempty intersection, i.e. $S_1 \cap S_2 \neq \emptyset.$
Intersecting families are intensively studied in extremal combinatorics, where they are the subject of many touchstone results, beginning with the seminal Erd\"{o}s-Ko-Rado theorem \cite{EKR61} and continuing to the present day.  Recent years have witnessed exciting progress on many problems dealing with intersecting families and their generalizations via analytic techniques that are highly relevant to the study of Boolean functions in theoretical computer science; see e.g.~\cite{Friedgut08,DF09,EKL19} and more generally \cite{Ellis22} for a recent and extensive survey.

Translating the above definition to the setting of Boolean functions, a function $f: \zo^n \to \zo$ is intersecting if the following ``pair condition'' holds: whenever $f(x)=f(y)=1$, there is (at least one) coordinate $i \in [n]$ such that $x_i=y_i=1.$ This is equivalent to
\begin{equation} \label{eq:intersecting}
x \leq \overline{y} \implies f(x) \leq \overline{f(y)},
\end{equation}
i.e. if $x \leq \overline{y}$, then having $f(y)=1$ implies that $f(x)$ must be 0, where $\overline{y}=(\overline{y}_1, \overline{y}_2, \dots, \overline{y}_n)$ is the  point in $\zo^n$ that is antipodal to $y.$ Finally, we observe that any $n$-variable Boolean function whose satisfying assignments all have first bit 1 is an intersecting function, so indeed the set of all $n$-variable intersecting Boolean functions is of doubly exponential size (at least $2^{2^{n-1}}$).

\medskip

\noindent {\bf Union-closedness.}
A set system ${\cal S} \subseteq 2^{[n]}$ is said to be \emph{union-closed} if whenever $S_1$ and $S_2$ belong to ${\cal S}$ then $S_1 \cup S_2$ also belongs to ${\cal S}$. In the Boolean function setting, this corresponds to the ``triple condition''  that $f: \zo^n \to \zo$ satisfy
\begin{equation} \label{eq:implication-UC}
z=x \cup y \implies 
f(x) f(y) \leq f(z),
\end{equation}
i.e.~if $f(x)=f(y)=1$ then $f(x \cup y)$ must also be 1.  Union-closed families have long been of interest in combinatorics, in part due to the well-known ``union-closed conjecture'' of Frankl \cite{Frankl95,BruhnSchaudt15}, which states that in any union-closed family some element $i \in [n]$ must appear in at least half the sets in the family.  Dramatic progress was recently made on the union-closed conjecture by Gilmer \cite{Gilmer22}, who proved a weaker form of the conjecture with $1/2$ replaced by $0.01$ (this constant was subsequently improved to ${\frac {3 - \sqrt{5}}{2}} \approx 0.38$ by \cite{AHS22,CL22,Pebody22,Sawin22}).  Since every monotone function is easily seen to be union-closed, union-closedness is a ``large'' property, with at least $2^{\Omega(2^n/\sqrt{n})}$ $n$-variable functions having the property.

\medskip

In this paper we initiate the study of intersectingness and union-closedness from a property testing perspective.  Given that (like monotonicity) these are  ``large'' properties that are defined by a simple ``pair'' or ``triple'' property, it is natural to wonder: Is the query complexity of testing these properties similar to the query complexity of testing monotonicity, or are these properties harder--- or easier---to test than monotonicity?

\subsection{Main Results}

As our main results, we show that both intersectingness and union-closedness are significantly more difficult to test than monotonicity: We give information-theoretic lower bounds which establish that neither of these properties admits a $\poly(n,1/\eps)$-query non-adaptive testing algorithm.  We also give sub-exponential non-adaptive testing algorithms for each of these properties; our algorithms have one-sided error (they never reject functions which have the property), while most of our lower bounds are for testing algorithms that are allowed two-sided error. 

We turn now to a detailed description of our main results.

\medskip

\noindent {\bf Positive results: Algorithms for Testing Intersectingness and Union-Closedness.} As a warm-up, and to develop intuition for these properties, we first give simple testing algorithms for intersectingness and for union-closedness which have sub-exponential query complexity:

\begin{theorem} [Testers for intersectingness and union-closedness]  \label{thm:main-alg}
There is a $\poly(n^{\sqrt{n \log(1/\eps)}},1/\eps)$-query, non-adaptive, one-sided\footnote{A tester is \emph{non-adaptive} if the choice of its $i$-th query point does not depend on the responses received to queries $1,\dots,i-1$.  A \emph{one-sided} tester for a class of functions is one which must accept every function in the class with probability 1.} algorithm for $\eps$-testing whether an unknown $f: \zo^n \to \zo$ is intersecting versus $\eps$-far from every intersecting function. 
The same is true for  union-closedness. 
\end{theorem}

\Cref{thm:main-alg} is proved by analyzing a ``pair tester'' for intersectingness and a  ``triple tester'' for union-closedness.  The distribution of pairs (respectively, triples) used by our algorithm is extremely simple, so it is natural to wonder whether a more sophisticated algorithm, perhaps using a cleverer distribution over pairs or triples, could result in a tester with an improved query complexity (indeed, this would be analogous to how the cleverer distribution over pairs used in \cite{CS13a,KMS18} resulted in a better query complexity for testing monotonicity than the simple distribution that was used in \cite{GGLRS}).  However, our main results---lower bounds for testing intersectingness and union-closedness---indicate that there are strong information-theoretic limitations on the possible performance of any non-adaptive testing algorithm for these properties.  We now describe our lower bound results.

\medskip

\noindent {\bf Negative Results: Lower Bounds for Testing Intersectingness and Union-Closedness.}  Our lower bounds show that both intersectingness and union-closedness are significantly harder to test than monotonicity: Neither of these properties has a $\poly(n,1/\eps)$-query non-adaptive testing algorithm, even if we allow two-sided error. (Recall that in contrast, the algorithms of \cite{GGLRS,CS13a,KMS18} are all $\poly(n,1/\eps)$-query non-adaptive one-sided testing algorithms for monotonicity.) In more detail, our main lower bound for intersectingness is the following (in all of our lower bound theorem statements, $c>0$ represents some sufficiently small absolute positive constant):

\begin{theorem} [Two-sided lower bound for intersectingness] \label{thm:main-lb-intersecting}
For $c > \eps\geq 1/\sqrt{n}$, any non-adaptive $\eps$-testing algorithm for whether an unknown $f: \zo^n \to \zo$ is intersecting versus $\eps$-far from intersecting must make $2^{\Omega(n^{1/4}/\sqrt{\eps})}$ queries to $f$.
\end{theorem}

When $\eps=1/\sqrt{n}$, the lower bound of \Cref{thm:main-lb-intersecting} essentially matches the performance of our algorithm from \Cref{thm:main-alg}, and even when $\eps$ is a constant, \Cref{thm:main-lb-intersecting} gives a $2^{\Omega(n^{1/4})}$ lower bound.
In view of the similarity between the defining conditions for monotonicity and intersectingness (\Cref{eq:mono} and \Cref{eq:intersecting}), we view \Cref{thm:main-lb-intersecting} as a potentially surprising result.

By imposing a stricter one-sided error condition, we can establish a stronger lower bound which almost matches the one-sided algorithm from \Cref{thm:main-alg} even for constant $\eps$:

\begin{theorem} [One-sided lower bound for intersectingness]\label{thm:one-sided-lb-intersecting}
For $c>\eps \geq 2^{-n}$, any non-adaptive one-sided $\eps$-testing algorithm for whether an unknown $f: \zo^n \to \zo$ is intersecting versus $\eps$-far from intersecting must make $2^{\Omega(\sqrt{n \log(1/\eps)})}$ queries to $f$.
\end{theorem}

Turning to union-closedness, the lower bound we give is not as strong as for intersectingness, but it is strong enough to rule out a $\poly(n,1/\eps)$-query non-adaptive algorithm, again even allowing two-sided error:

\begin{theorem} [Two-sided lower bound for union-closedness]\label{thm:two-sided-lb-UC}
For $c>\eps \geq 2^{-n^{0.49}}$, any non-adaptive $\eps$-testing algorithm for whether an unknown $f: \zo^n \to \zo$ is union-closed versus $\eps$-far from union-closed must make $n^{\Omega(\log(1/\eps))}$ queries to $f$.
\end{theorem}

As we discuss in \Cref{sec:discussion}, an interesting goal for future work is to narrow the gap between our algorithm and our lower bound for testing union-closed families.

\subsection{Techniques}

In this section, we give a technical overview of our main results, starting with the lower bounds. 

\medskip

\noindent {\bf Lower Bounds.}
Our two-sided lower bound for intersectingness, \Cref{thm:main-lb-intersecting}, builds on a lower bound approach for \emph{tolerant} monotonicity testing which was introduced in \cite{PRW22} and was recently quantitatively strengthened in \cite{CDLNS23}.
  As is standard for non-adaptive property testing lower bounds, \cite{PRW22} and \cite{CDLNS23} use Yao's minimax lemma and define  a ``yes''-distribution $\Dyes$  and a ``no''-distribution $\Dno$ over Boolean functions; in the rest of this discussion we focus chiefly on \cite{CDLNS23}.  A function $\boldf$ drawn from either of the \cite{CDLNS23} distributions $\Dyes$ or $\Dno$ is defined based on a random partition of the $n$ variables into a (large) set of ``control'' variables and a (small) set of ``action'' variables.  In both cases $\boldf \sim \Dyes$ or $\boldf \sim \Dno$, the definition of $\boldf$ involves a ``Talagrand DNF,'' $\bT=\bT_1 \vee \cdots \vee \bT_m,$ which is essentially a random monotone DNF formula  over the control variables.\footnote{The earlier work \cite{PRW22} used a different function over the control variables instead of a Talagrand DNF.}  The crucial assignments to $\boldf$ are the ones for which the control variables satisfy exactly one term $\bT_i$ of the Talagrand DNF; for such an input string $x$, the value of $\boldf$ then depends on the setting of the action variables, and the difference between $\boldf \sim \Dyes$ and $\boldf \sim \Dno$ comes from how the function is defined over the action variables in each case.  The values of the function on the action subcubes are carefully defined in such a way as to make it impossible for a testing algorithm to distinguish a “yes”-function from a “no”-function unless  it manages to query two inputs $x, x'$ which (i) both have their control variables set in such a way as to uniquely satisfy the same term $\bT_i$, but (ii) differ on ``many coordinates'' among the action variables:  essentially, one of $x,x'$ must have its vector of action bits landing in the top portion of the action subcube while the other one must have its vector of action bits landing in the bottom portion.  The crux of the non-adaptive lower bound of \cite{CDLNS23} is the tension between requirements (i) and  (ii): if  $x$ and $x'$ differ in too  many coordinates then it is difficult to satisfy (i), but if they differ in too few coordinates then it is difficult to satisfy (ii).

In the setting of monotonicity  testing, the  \cite{CDLNS23} construction's yes-functions are only close to, but not actually, monotone; their non-monotonicity essentially comes from assignments for which the vector of action bits lands in the middle portion of the action subcube.  This is why the  mildly exponential lower bound proved in  that paper only holds for \emph{tolerant} monotonicity  testing (indeed, the  existence of highly efficient monotonicity  testers  \cite{GGLRS,CWX17stoc,KMS18} implies that quantitatively strong lower bounds such as those of \cite{CDLNS23} are  impossible for  ``standard'' non-tolerant monotonicity testing).  The main component of our lower bound for intersectingness in this paper is a careful modification of the \cite{CDLNS23} construction; we show that, perhaps surprisingly, for the modification that we introduce, the yes-functions have satisfying assignments which form a \emph{perfectly} intersecting family,  while the no-functions are far from intersecting. We thus obtain a quantitatively strong lower bound, similar to \cite{CDLNS23}, already for the ``standard'' testing problem of intersectingness rather than the more challenging tolerant version.

Our $2^{\Omega(\sqrt{n \log(1/\eps)})}$-query one-sided lower bound for intersectingness, \Cref{thm:one-sided-lb-intersecting}, takes a related but somewhat simpler approach.  In a nutshell, since for one-sided lower bounds it is not necessary to give a yes-distribution and establish indistinguishability of yes-functions and no-functions, it turns out that we can dispense with the Talagrand DNF part of the construction. Instead, our construction ``hides'' a randomly chosen ``small'' set of action bits in a simpler way (see \Cref{sec:intersecting-one-sided} for details); since we do not need to use the Talagrand DNF, it turns out that we can have the ``small'' set of action bits be larger than in our intersectingness lower bound, and this lets us obtain a quantitatively stronger lower bound.

Finally, our $n^{\Omega(\log(1/\eps))}$-query two-sided lower bound for union-closedness, like our two-sided intersectingness lower bound, uses the framework of control bits and action bits with a Talagrand DNF over the control bits. This construction uses a somewhat different definition of the yes- and no- functions over the action bits, which now ensures that a testing algorithm can distinguish yes-functions from no-functions only if it manages to query two inputs whose control variables satisfy the same term $\bT_i$ but whose action variables are set to two particular \emph{antipodal} assignments in the action cube.   For this construction we use many fewer action bits than in the earlier construction (and the quantitative lower bound obtained is correspondingly weaker than the lower bound of the earlier construction); this is because in our no-functions, the distance to union-closedness is inverse exponential in the dimension of the action cubes. 

\medskip

\noindent {\bf Algorithms.} Our algorithms for testing intersectingness and for testing union-closedness are similar at a high level; for conciseness we only describe the algorithm for testing union-closedness.  

As is standard for testing algorithms, we consider the two possible scenarios. In the ``yes" case, the given function $f$ is union-closed. In the ``no" case, the function $f$ is $\epsilon$-far in Hamming distance from any union-closed function.

At a conceputal level, the first simplification is as follows: given $f$, we can define a truncated version of $f$, call it $f_{\mathsf{trunc}}$ as follows: for any $x$ such that $| x | \in  [n/2 - T, n/2 +T]$ where $T = \sqrt{n \log(4/\epsilon)}$, $ f_{\mathsf{trunc}}(x) = f(x)$. If $| x | > n/2 + T$, we set $ f_{\mathsf{trunc}}(x) = 1$ and if $| x | < n/2 - T$, we set $f_{\mathsf{trunc}}(x)=0$. In other words, $f_{\mathsf{trunc}}$ is obtained by keeping it the same as $f$ in the middle $2T$ layers; otherwise, it is set to $1$ in the layers above the middle layers and $0$ below it. Since all but $\epsilon/2$ fraction of the mass of the discrete cube lies in the layers $[n/2-T, n/2+T]$, the following is immediate: (i) if $f$ is union-closed, so is $f_{\mathsf{trunc}}$; (ii) if $f$ is $\epsilon$-far from union-closed, $f_{\mathsf{trunc}}$ is also $\epsilon/2$-far from union-closed (\Cref{claim:zeroing}). 
The above property of $f_{\mathsf{trunc}}$ ensures that instead of working with $f$, the algorithm can instead work with $f_{\mathsf{trunc}}$.

 Now, the main idea behind the algorithm is to search for {\em violations of union-closedness}. In this sense, our algorithm is similar in spirit to algorithms for monotonicity testing~\cite{GGLRS, CS13a, KMS18} which search for violations of monotonicity. In particular, we call a sequence $(x_1, \ldots, x_k, x_1 \cup \ldots \cup x_k)$ a {\em union-closed violating tuple} if $f(x_1) = \ldots = f(x_k)=1$ and $f(x_1 \cup \ldots \cup x_k)=0$ -- we will abbreviate this as a $\mathsf{UC}$-violating tuple.  Note that if the algorithm finds a union-closed violating tuple in $f$, then it is a certificate for $f$ not being union-closed. 
 
 The main technical lemma we prove is that if  $f$ is $\epsilon$-far from union closed, then it has at least $\epsilon \cdot 2^n$ $\mathsf{UC}$-violating tuples which are end-disjoint. This means that for any two such  tuples $(x_1, \ldots, x_k, x_1 \cup \ldots \cup x_k)$ and $(y_1, \ldots, y_k, y_1 \cup \ldots \cup y_k)$, the last coordinate $(x_1 \cup \ldots \cup x_k) \not = (y_1 \cup \ldots \cup y_k)$. The proof of this lemma 
is quite simple -- essentially, we show that the function $f$ can be changed to a union closed function by only modifying it at points which are the last coordinate of a $\mathsf{UC}$-violating tuple. Given this lemma, it follows that $f$ must have at least $\epsilon \cdot 2^n$ end-disjoint $\mathsf{UC}$-violating tuples. Since $f$ and $f_{\mathsf{trunc}}$ are $\epsilon/2$-close to each other, it follows that $f_{\mathsf{trunc}}$ also has at least $\epsilon/2 \cdot 2^n$ end-disjoint $\mathsf{UC}$-violating tuples.

We next observe that a $\mathsf{UC}$-violating tuple 
$(x_1, \ldots, x_k, x_1 \cup \ldots \cup x_k)$ 
for 
$f_{\mathsf{trunc}}$ is such that
(i) for each $1 \le i \le k$,  $|| x_i | -n/2| \le T$; (ii) $| | x_1\cup \ldots \cup x_k |-n/2| \le T$. Let us call a point $x=x_1 \cup \ldots \cup x_k$ a witness if there is a $\mathsf{UC}$-violating tuple 
$(x_1, \ldots, x_k, x_1 \cup \ldots \cup x_k)$ satisfying the above conditions. 
From the fact that $f_{\mathsf{trunc}}$ also has at least $\epsilon/2 \cdot 2^n$ end-disjoint $\mathsf{UC}$-violating tuples, it follows that there are at least $\epsilon/2 \cdot 2^n$ points which are a witness. 

 Our algorithm now proceeds as follows: We sample a random point $\bx \in \{0,1\}^n$ conditioned on $|| \bx | - n/2| \le T$. Next, we query $f$ on $\bx$ as well as all the points in the set $\bx_{\downarrow} := \{y \le \bx: || y |  - n/2| \le T\}$. We then check if there are any points $y_1, \ldots, y_k \in \bx_{\downarrow}$ such that $(y_1, \ldots, y_k, \bx)$ is a $\mathsf{UC}$-violating tuple. Note that if $f$ is union-closed, then the algorithm is certainly not going to find a $\mathsf{UC}$-violating tuple, i.e., it has perfect completeness. On the other hand, if $f$ is at least $\epsilon$-far from union closed, then the point $\bx$ sampled above is a witness with probability $\epsilon/2$. If $\bx$ is a witness then since we are querying every point in $\bx_{\downarrow}$, the algorithm is going to find a $\mathsf{UC}$-violating tuple. 
 
 Thus, repeating the above procedure say $100/\epsilon$ times, the algorithm will still have perfect completeness. On the other hand, if $f$ is $\epsilon$-far from union-closed, it is going to find a $\mathsf{UC}$-violating tuple with probability at least $0.9$. The query complexity of the algorithm is given by $O(1/\epsilon) \cdot|\bx_{\downarrow}|$. As $|\bx_{\downarrow}|$ is uniformly bounded by $n^{O(\sqrt{n \log (1/\epsilon)})}$, this establishes the upper bound on the query complexity of our algorithm.  (While the algorithm described above is not a ``triple tester,'' an easy modification of the algorithm and its analysis yields a triple tester with similar query complexity; see \Cref{sec:ubs}.)

\subsection{Related Work}

As mentioned earlier, some of the technical specifics of our lower bound constructions build off of the tolerant testing lower bounds of \cite{PRW22} and \cite{CDLNS23}; in particular, the idea, first introduced by \cite{PRW22}, of ``hiding'' a set of action variables among the entire set of input variables was a significant influence on the lower bound constructions of the current paper.  More generally, the entire broad literature on monotonicity testing of Boolean functions (i.e.~testing upward-closed set systems) provided the conceptual backdrop for a study of the testability of other types of combinatorial finite set systems.

We note that the recent work of Filmus et al.~\cite{FLMM20} (see also \cite{CFMMS22}) studies the problem of ``AND-testing,'' which at first glance may seem to be related to the problems we consider.  The ``AND-property'' is that of satisfying the implication
\begin{equation} \label{eq:and-testing}
z = x \cap y \implies f(z) = f(x) \wedge f(y)
\end{equation}
for \emph{every} $x,y \in \zo^n$; the
main result of \cite{FLMM20}, roughly speaking, is that the only functions which have a high probability of satisfying \Cref{eq:and-testing}
for uniform random $x,y$ are functions which are close to being either a constant-function or an AND of some subset of the $n$ input variables.  

Despite the superficial resemblance between \Cref{eq:implication-UC} and \Cref{eq:and-testing}, it turns out that the AND-property and the properties we consider are of quite different character from each other.  To see this, observe that the only functions $f: \zo^n \to \zo$ which perfectly satisfy the AND-property are constant functions and AND-functions; hence there are only $O(2^n)$ many possible yes-functions, and every yes-function must have a very precise and rigid structure (and a very simple description). This is quite different from the intersectingness and union-closedness properties we study; each of these properties has $2^{2^{\Theta(n)}}$ many yes-functions, and hence yes-functions do not need to be so highly structured (and by standard counting arguments almost all yes-functions require highly complex descriptions).  As another point of difference, the \cite{FLMM20} result mentioned above implies that there is an $O_\eps(1)$-query non-adaptive one-sided tester for the AND-property. In contrast, our \Cref{thm:two-sided-lb-UC} shows that even two-sided non-adaptive testers for the property of union-closedness must have a query complexity which not only depends on $n,$ but in fact is at least $n^{\Omega(\log(1/\eps))}.$


\section{Preliminaries}
\label{sec:prelims}

We will write 
\[{[n]\choose k} :=  \cbra{S\sse[n] : |S| = k}\]
to denote the collection of all $k$-element subsets of $[n]$, and for a subset $I \subseteq [n]$ we will write ${[n]\choose I}$ to denote $\cup_{j \in I} {[n]\choose j}$. 
We will denote the $0/1$-indicator of an event $A$ by $\mathbf{1}\cbra{A}$. All probabilities and expectations will be with respect to the uniform distribution over the relevant domain unless stated otherwise. We use boldfaced letters such as $\bx, \boldf$, and $\bA$ to denote random variables (which may be real-valued, vector-valued, function-valued, or set-valued; the intended type will be clear from the context).
We write $\bx \sim \calD$ to indicate that the random variable $\bx$ is distributed according to probability distribution $\calD$. 

\begin{notation}
	Given a string $x\in\zo^n$ and a set $A\sse[n]$, we write $x_A \in \zo^{A}$ to denote the $|A|$-bit string obtained by restricting $x$ to coordinates in $A$, i.e. $x_A := (x_i)_{i\in A}$, and we write $|x|$ to denote the number of 1's in $x$.
\end{notation}

We will frequently  view strings in $\zo^n$ as subsets of $[n]$ and vice versa; i.e.~for $x,y \in \zo^n$ we refer to ``$x \cap y$'' to mean the string in $\zo^n$ which has a 1 in coordinate $i$ iff $x_i=y_i=1.$

Given two Boolean functions $f,g\isazofunc$, we define the \emph{distance} between $f$ and $g$ (denoted by $\dist(f,g)$) to be the normalized Hamming distance between $f$ and $g$, i.e. 
\[\dist(f,g) := \Prx_{\bx\sim\zo^n}\sbra{f(\bx)\neq g(\bx)}.\]
A \emph{property} $\calP$ is a collection of Boolean functions; we say that a function $f\isazofunc$ is \emph{$\eps$-far from the property $\calP$} if 
\[\dist(f,\calP) := \min_{g\in\calP} \dist(f,g)\geq \epsilon.\]

\subsection{Lower Bounds for Testing Algorithms}
\label{subsec:testing-lb-prelims}

Our query-complexity lower bounds for testing algorithms are obtained via Yao's minimax principle~\cite{Yao:77}, which we recall below. (We remind the reader that an algorithm for the problem of $\eps$-property testing is correct on an input function $f$ provided that it outputs ``yes'' if $f$ \emph{perfectly} satisfies the property and outputs ``no'' if $f$ is $\eps$-far from the property; if the distance to the property is strictly between $0$ and $\eps$ then the algorithm is correct regardless of what it outputs.)
\begin{theorem}[Yao's principle] \label{thm:yao-minimax}

	To prove a $q$-query lower bound on the worst-case query complexity of any non-adaptive {randomized} testing algorithm, it suffices to give a distribution $\calD$ on instances
	such that for any $q$-query non-adaptive \emph{deterministic} algorithm $\calA$, we have 
	\[\Prx_{\bm{f}\sim \calD}\big[\calA\text{ is correct on }\bm{f}\big]\leq 99.9\%.\] 
	Here $99.9\%$ can be replaced by any universal constant in $[0,1)$. 
\end{theorem}

\subsection{Talagrand's Random DNF} 
\label{subsec:kane-and-talagrand}

We define a useful distribution over Boolean functions that will play a central role in the proofs of our lower bounds. The construction is a slight generalization of a distribution over DNF (disjunctive normal form) formulas that was constructed by Talagrand~\cite{Talagrand:96}.  The generalization we consider, which was also studied in \cite{CDLNS23}, is that we allow a parameter $\eps$ to control the size of each term and the number of terms; the original construction corresponds to $\eps=1$.

\begin{definition}[Talagrand's random DNF] \label{def:talagrand}
Let $\eps\in (0,1]$ and let $L\coloneqq 0.1\cdot2^{\sqrt{n}/\epsilon}$. Let 
$\Tal(n, \epsilon)$ 
be the following distribution on ordered tuples of $L$ monotone terms: for each $i=1,\dots,L$, the $i$-th term is obtained by independently drawing a set $\bT_i\sse [n]$ where each set $\bT_i$ is obtained by drawing $\sqrt{n}/\eps$ elements of $[n]$ independently and with replacement.  We use $\bm{T}$ to denote the ordered tuple $\bm{T}=(\bm{T}_1,\cdots,\bm{T}_L)$ which is a draw from $\Tal(n, \epsilon)$.
		Then a ``Talagrand DNF'' is given by 
		\[\boldf(x) = \bigvee_{\ell=1}^{L} \pbra{\bigwedge_{j\in \bT_{\ell}} x_{j}}.\]
\end{definition}

It is clear that any Talagrand DNF obtained by a draw from $\Tal(n, \epsilon)$ is a monotone function. 

We will frequently view $T_{i}\subseteq [n]$ as the term $\bigwedge_{j\in T_{i}} x_{j}$, where we say $T_i(x)=1$ if and only if $x_{j}=1$ for all $j\in T_{i}$. We may also write $T=(T_1,\cdots,T_k)$ to represent a DNF, which is defined by the disjunction of the terms $T_{i}$. We will often be interested in the probability of a random input $\bx\sim\zo^n$ satisfying a unique term $\bT_i$ in a Talagrand DNF; towards this, we introduce the following notation:

\begin{notation} \label{notation:term-count}
	Given a DNF $T=(T_1,\cdots,T_k)$ where each $T_i$ is a term, we define the collection of \emph{terms of $T$ satisfied by $x$}, written $S_T(x)$, as 
	\[S_T(x) := \cbra{\ell\in [k] : T_\ell(x) = 1}.\]
\end{notation}

The following claim shows that on average over the draw of $\bm{T}\sim\Tal(n,\eps)$, an $\Omega(\eps)$ fraction of strings from $\zo^n$ satisfy a \emph{unique} term in the Talagrand DNF (i.e. $|S_{\bm{T}}(x)| = 1$ for $\Omega(\eps)$-fraction of $x\in\zo^n$). We note that an elegant argument of Kane~\cite{Kane13monotonejunta} gives this for $\eps=\Theta(1)$, but this argument does not extend to the setting of small $\eps$ which we require. The proof below appears in \cite{CDLNS23}; we repeat the argument here for the reader's convenience.

\begin{proposition} \label{prop:talagrand-unique-property}
For $\eps\in (0,1]$, let $\bm{T}\sim\Tal(n,\eps)$ be as in \Cref{def:talagrand}. Then 
	\[\Prx_{\bm{T}, \bx}\sbra{|S_{\bm{T}}(\bx)| = 1} = \Omega\left(\max \{\eps,1/\sqrt{n}\}\right).\]
\end{proposition}

\begin{proof}
	Note that \Cref{prop:talagrand-unique-property} is immediate if the following holds: For every string $x\in\zo^n$ with $|x| \in [n/2, n/2 + 0.05\eps\sqrt{n}]$,\footnote{Note that when $0.05\eps\sqrt{n}<1$, only strings $x\in \{0,1\}^n$ with $|x|=n/2$ are considered.} we have 
	\begin{equation} \label{eq:tal-goal}
		\Prx_{\bm{T}}\sbra{|S_{\bm{T}}(x)| = 1} = \Omega(1).
	\end{equation}
	This is because a straightforward application of the Chernoff bound, and the well-known middle binomial coefficient bound ${n \choose n/2}/2^n = \Theta(1/\sqrt{n})$, together imply that 
	\[\Prx_{\bx}\sbra{|x| \in [n/2, n/2 + 0.05\eps\sqrt{n}]} = \Omega\left(\max\{\epsilon,1/\sqrt{n}\}\right).\]
	
We prove (\ref{eq:tal-goal}) in the rest of the proof.
Fix such an $x\in\zo^n$ and let $\bT_i$ be one of the $0.1\cdot2^{\sqrt{n}/\epsilon}$ terms of $\bm{T}\sim\Tal(n,\eps)$. Recalling that $\bT_i$ consists of $\sqrt{n}/\eps$ many variables (with repetition), we have 
	\begin{align*}
		\Prx_{\bT_i}\sbra{\bT_i(x) = 1} &\leq \pbra{\frac{1}{2} + \frac{0.05\epsilon}{\sqrt{n}}}^{\sqrt{n}/\epsilon} =\pbra{\frac{1}{2} + \frac{0.1\epsilon}{2\sqrt{n}}}^{\sqrt{n}/\epsilon} \leq 2^{-\sqrt{n}/\epsilon}\exp\pbra{0.1}.
	\end{align*}
	where in the second inequality we used the fact that $1+x\leq e^x$ for all $x\in\R$. It follows by the linearity of expectation that for any $x$ as above, we have
	\[\Ex_{\bm{T}}\sbra{|S_{\bm{T}}(x)|} = \sum_{i=1}^{0.1\cdot2^{\sqrt{n}/\epsilon}} \Prx_{\bT_i}\sbra{\bT_i(x) = 1} \leq 0.1\cdot\exp(0.1) < 0.12.\]
	Markov's inequality then implies that 
	\begin{equation} \label{eq:tal-proof-markov}
		\Prx_{\bm{T}}\sbra{|S_{\bm{T}}(x)| \geq 2}\leq \frac{1}{2}\cdot\Ex_{\bm{T}}\sbra{|S_{\bm{T}}(x)|} \leq 0.06.
	\end{equation}
On the other hand, since $|x| \geq n/2$, we have 
	\begin{align} 
		\Prx_{\bm{T}}\sbra{|S_{\bm{T}}(x)| = 0} 
		\le \left(1-\left(\frac{1}{2}\right)^{\sqrt{n}/\epsilon}\right)^{0.1\cdot 2^{\sqrt{n}/\epsilon}}
		\le \exp(-0.1)<0.91,\label{eq:tal-proof-unsat-prob}
	\end{align}
	where the second inequality used again $1+x\le e^x$.
	Combining \Cref{eq:tal-proof-markov,eq:tal-proof-unsat-prob}, we get that 
	\[\Prx_{\bm{T}}\sbra{|S_{\bm{T}}(x)| = 1} >0.03,\]
	thus establishing \Cref{eq:tal-goal}, which in turn completes the proof.
\end{proof}


\def\Dy{\Dyes}
\def\Dn{\Dno}
\def\prob{\textrm{Pr}}
\def\ManyOne{\textsf{ManyOne}}
\def\GoodTalagrand{\textsf{GoodTalagrand}}
\def\Bad{\textsf{Bad}}

\section{Lower Bounds for Testing Intersecting Families}
\label{sec:intersecting}

We now present our lower bound for two-sided non-adaptive testers for intersecting families.
 As mentioned earlier, the construction builds closely on the earlier constructions of \cite{PRW22,CDLNS23} which were used in those papers for tolerant testing lower bounds.

Let $\eps\in(0,c]$ be a parameter with $c>\eps\geq c_0/\sqrt{n}$ for some sufficiently large constant $c_0$ and sufficiently small constant $c>0$. We start with some objects that we need in the construction of the two distributions $\Dyes$ and $\Dno$. 
We partition the variables $x_1,\cdots,x_n$ into \emph{control} variables and \emph{action} variables as follows: Let $a\coloneqq \sqrt{n}/\eps$ and let ${A}\subseteq [n]$ be a fixed subset of $[n]$ of size $a$. Let $C\coloneqq [n]\setminus A$. We refer to the variables $x_i$ for $i\in C$ as \emph{control variables} and the variables $x_i$ for $i\in A$ as \emph{action variables}. We first define two pairs of functions over $\{0,1\}^A$ on the action variables as follows (we will use these functions later in the definition of $\Dy$ and $\Dn$):

\begin{equation*}
g^{(+,0)}(x_A) =
	\begin{cases}
       \ 0 & |x_A|>\frac{a}{2}+\sqrt{a}; \\[0.5ex]
       \ 0 & |x_A|\in[\frac{a}{2}-\sqrt{a},\frac{a}{2}+\sqrt{a}]; \\[0.5ex]
       \ 0 & |x_A|<\frac{a}{2}-\sqrt{a}.
    \end{cases}
\quad g^{(+,1)}(x_A) =
	\begin{cases}
       \ 1 & |x_A|>\frac{a}{2}+\sqrt{a}; \\[0.5ex]
       \ 0 & |x_A|\in[\frac{a}{2}-\sqrt{a},\frac{a}{2}+\sqrt{a}]; \\[0.5ex]
       \ 1 & |x_A|<\frac{a}{2}-\sqrt{a}.
    \end{cases}
\end{equation*}
and
\begin{equation*}
g^{(-,0)}(x_A) =
	\begin{cases}
       \ 1 & |x_A|>\frac{a}{2}+\sqrt{a}; \\[0.5ex]
       \ 0 & |x_A|\in[\frac{a}{2}-\sqrt{a},\frac{a}{2}+\sqrt{a}]; \\[0.5ex]
       \ 0 & |x_A|<\frac{a}{2}-\sqrt{a}.
    \end{cases}
\quad g^{(-,1)}(x_A) =
	\begin{cases}
       \ 0 & |x_A|>\frac{a}{2}+\sqrt{a}; \\[0.5ex]
       \ 0 & |x_A|\in[\frac{a}{2}-\sqrt{a},\frac{a}{2}+\sqrt{a}]; \\[0.5ex]
       \ 1 & |x_A|<\frac{a}{2}-\sqrt{a}.
    \end{cases}
\end{equation*}

Now we are ready to define the distributions $\Dy$ and $\Dn$ over $f:\{0,1\}^{n+2}\rightarrow\{0,1\}$. We follow the convention that random variables are in boldface and fixed quantities are in the standard typeface.

A function $\fyes\sim \Dy$ is drawn as follows.
We start by sampling a subset $\bm{A}\subseteq [n]$ of size $a$ uniformly at random and let $\bm{C}\coloneqq [n]\setminus \bm{A}$. Note that there are in total $n-a$ control variables. We let $L\coloneqq 0.1\cdot2^{\sqrt{n-a}/\epsilon}$ and draw an $L$-term monotone Talagrand DNF $\bm{T}\sim \Tal(n-a,\eps)$ on $\bm{C}$ as described in \Cref{def:talagrand}. Finally, we sample $L$ random bits $\bm{b}\in\{0,1\}^L$ uniformly at random. 
Given $\bA,\bT$ and $\bb$, $\fyes$ is defined 
  by letting  $\fyes(x,0,0)=\fyes(x,1,1)=0$ for all $x\in\{0,1\}^n$, and letting

\begin{align*}
\fyes(x,0,1) &=
	\begin{cases}
       \ 0 & |S_{\bm{T}}(x_{\bm{C}})|\neq 1; \\[0.5ex]
       \ g^{(+,0)}(x_{\bm{A}}) & S_{\bm{T}}(x_{\bm{C}})=\{\ell\} \text{ and } \bm{b}_{\ell}=0; \\[0.5ex]
       \ g^{(+,1)}(x_{\bm{A}}) & S_{\bm{T}}(x_{\bm{C}})=\{\ell\} \text{ and } \bm{b}_{\ell}=1.
    \end{cases}\\[1ex]
 \fyes(x,1,0) &=
	\begin{cases}
       \ 0 & |S_{\bm{T}}(\overline{x}_{\bm{C}})|\neq 1; \\[0.5ex]
       \ g^{(+,1)}(x_{\bm{A}}) & S_{\bm{T}}(\overline{x}_{\bm{C}})=\{\ell\} \text{ and } \bm{b}_{\ell}=0; \\[0.5ex]
       \ g^{(+,0)}(x_{\bm{A}}) & S_{\bm{T}}(\overline{x}_{\bm{C}})=\{\ell\} \text{ and } \bm{b}_{\ell}=1.
    \end{cases}
\end{align*}
(Recall that $\overline{x}$ is the bitwise complement of string $x$).

To draw a function $\fno\sim \Dn$, we sample $\bA,\bT$ and $\bb$ exactly as in the definition of $\Dy$ above, but we use $g^{(+,b)}$ and $g^{(-,b)}$ functions in a different way than in the $\Dy$ functions described above. In more detail, $\fno$ is defined by 
  $\fno(x,0,0)=\fno(x,1,1)=0$ for all $x\in\{0,1\}^n$, and  
\begin{align*}
\fno(x,0,1) &=
	\begin{cases}
       \ 0 & |S_{\bm{T}}(x_{\bm{C}})|\neq 1; \\[0.5ex]
       \ g^{(-,0)}(x_{\bm{A}}) & S_{\bm{T}}(x_{\bm{C}})=\{\ell\} \text{ and } \bm{b}_{\ell}=0;\\[0.5ex]
       \ g^{(-,1)}(x_{\bm{A}}) & S_{\bm{T}}(x_{\bm{C}})=\{\ell\} \text{ and } \bm{b}_{\ell}=1.
    \end{cases}\\
\fno(x,1,0) &=
	\begin{cases}
       \ 0 & |S_{\bm{T}}(\overline{x}_{\bm{C}})|\neq 1; \\[0.5ex]
       \ g^{(-,0)}(x_{\bm{A}}) & S_{\bm{T}}(\overline{x}_{\bm{C}})=\{\ell\} \text{ and } \bm{b}_{\ell}=0; \\[0.5ex]
       \ g^{(-,1)}(x_{\bm{A}}) & S_{\bm{T}}(\overline{x}_{\bm{C}})=\{\ell\} \text{ and } \bm{b}_{\ell}=1.
    \end{cases}
\end{align*}

See \Cref{fig:int-yes,fig:int-no} for illustrations of the yes- and no- functions.


\begin{figure} \label{fig:yes-intersecting}
\centering

\begin{tikzpicture}

\node (potato1) at (0,5.5) {\begin{tikzpicture}[scale=0.9]

	
	\fill[color=purple!20!white!80] (0, -3) -- (-3, 0) -- (0, 3) -- (3, 0);
	\fill[color=white] (-2.75, 0.25) -- (-1, -1.5) -- (1.75, 1.25) -- (0, 3);\fill[pattern=crosshatch,pattern color=black, opacity=0.175] (-2.75, 0.25) -- (-1, -1.5) -- (1.75, 1.25) -- (0, 3);
	\fill[color=white] (-2.25, 0.75) -- (0, -1.5) -- (2.25, 0.75) -- (0, 3);\fill[pattern=crosshatch,pattern color=black, opacity=0.175] (-2.25, 0.75) -- (0, -1.5) -- (2.25, 0.75) -- (0, 3);
	\fill[color=white] (-1.75, 1.25) -- (1, -1.5) -- (2.75, 0.25) -- (0,3) -- (-1.75, 1.25);\fill[pattern=crosshatch,pattern color=black, opacity=0.175] (-1.75, 1.25) -- (1, -1.5) -- (2.75, 0.25) -- (0,3) -- (-1.75, 1.25);
	
	\fill[color=purple!20!white!80] (-2.25, 0.75) -- (-0.5, -1) -- (0, -0.5) -- (0.5, -1) -- (2.25, 0.75) -- (0, 3);
	
	\draw[-] (0, -3) -- (-3, 0) -- (0, 3) -- (3, 0) -- (0, -3);

	\draw[-] (-2.75, 0.25) -- (-1, -1.5) -- (1.75, 1.25);
	\draw[-] (-2.25, 0.75) -- (0, -1.5) -- (2.25, 0.75);
	\draw[-,line width=0.3mm] (-1.75, 1.25) -- (1, -1.5) -- (2.75, 0.25) -- (0,3) -- (-1.75, 1.25);
		
	\node (0) at (0, -2){$0$};
	\node (1) at (0, 1){$0$};
	\node (talagrand-cube) at (1.15, -0.85){$\bT_{\ell}$};

	
	\fill[color=purple!20!white!80] (9, 1) -- (7.5, 2.5) -- (9, 4) -- (10.5, 2.5);
	\node (heads) at (5.5, 2.5){\includegraphics[width=1.1cm]{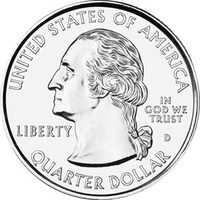}};
	\node (bh) at (5.5, 1.5){If $\bb_{\ell} = 1$:};
	\fill[color=purple!20!blue!40!white!70] (10, 3) -- (9, 4) -- (8, 3);
	\fill[color=purple!20!blue!40!white!70] (10, 2) -- (9, 1) -- (8, 2);
	\draw[] (10, 2) -- (8, 2);
	\draw[] (10, 3) -- (8, 3);
	
	\node() at (9, 3.4){$1$};
	\node() at (9, 2.5){$0$};
	\node() at (9, 1.6){$1$};
	
	\draw[-] (9, 1) -- (7.5, 2.5) -- (9, 4) -- (10.5, 2.5) -- (9, 1);
	
	\fill[color=purple!20!white!80] (9, -3.5) -- (7.5, -2) -- (9, -0.5) -- (10.5, -2);
	\node (tails) at (5.5, -2){\includegraphics[width=1.1cm]{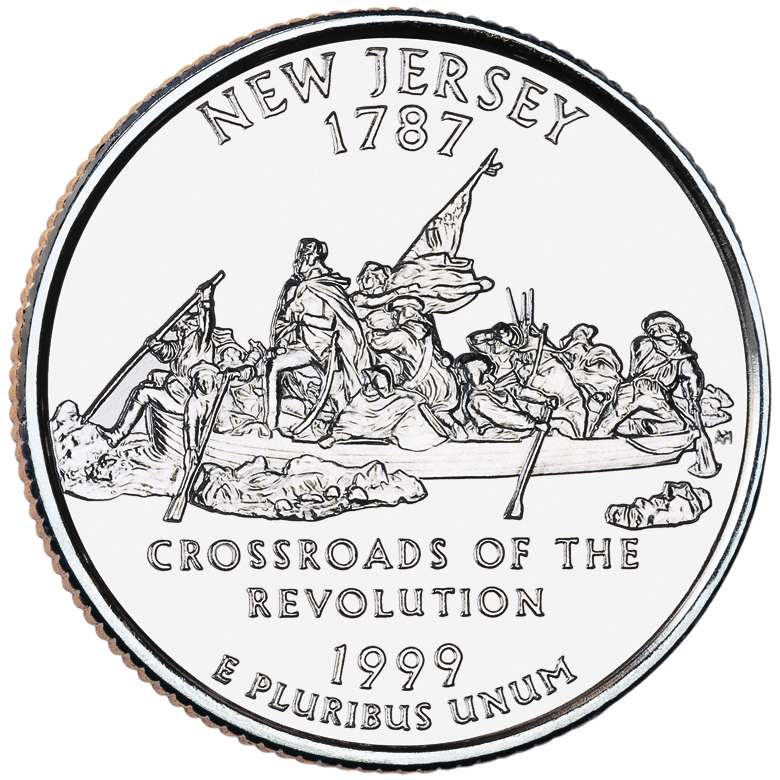}};
	\node (bt) at (5.5, -3){If $\bb_{\ell} = 0$:};
	
	\draw[-] (9, -3.5) -- (7.5, -2) -- (9, -0.5) -- (10.5, -2) -- (9, -3.5);	
	\draw[] (10, -2.5) -- (8, -2.5);
	\draw[] (10, -1.5) -- (8, -1.5);
	
	\node() at (9, -1.1){$0$};
	\node() at (9, -2){$0$};
	\node() at (9, -2.9){$0$};

	
	\draw[-latex,dashed,color=gray] (2.35, 0.1) -- (4.5, 2.25);
	\draw[-latex,dashed,color=gray] (2.35, 0.1) -- (4.5, -2.25);
	\node () at (2.35, 0.1) [circle,fill,inner sep=1pt]{};
		\node () at (2, 0.1) {${x}_{\bC}$};
	
	\node (control) at (0, -4.75){\phantom{$\zo^{\bC} \equiv \zo^m$}};
	\node (action) at (9, -4.75){\phantom{$\zo^{\bA} \equiv \zo^{a}$}};

\end{tikzpicture}};

\node (potato2) at (0,-6) {\begin{tikzpicture}[scale=0.9]

	
	\fill[color=purple!20!white!80] (0,-3) -- (-3,0) -- (0,3) -- (3,0);
	\fill[color=white] (-2.75, -0.25) -- (-1, 1.5) -- (1.75, -1.25) -- (0, -3);\fill[pattern=crosshatch,pattern color=black, opacity=0.175] (-2.75, -0.25) -- (-1, 1.5) -- (1.75, -1.25) -- (0, -3);
	\fill[color=white] (-2.25, -0.75) -- (0, 1.5) -- (2.25, -0.75) -- (0, -3);\fill[pattern=crosshatch,pattern color=black, opacity=0.175] (-2.25, -0.75) -- (0, 1.5) -- (2.25, -0.75) -- (0, -3);
	\fill[color=white] (-1.75, -1.25) -- (1, 1.5) -- (2.75, -0.25) -- (0,-3) -- (-1.75, -1.25);\fill[pattern=crosshatch,pattern color=black, opacity=0.175] (-1.75, -1.25) -- (1, 1.5) -- (2.75, -0.25) -- (0,-3) -- (-1.75, -1.25);
	
	\fill[color=purple!20!white!80] (-2.25, -0.75) -- (-0.5, 1) -- (0, 0.5) -- (0.5, 1) -- (2.25, -0.75) -- (0, -3);
	
	\draw[-] (0, -3) -- (-3, 0) -- (0, 3) -- (3, 0) -- (0, -3);

	\draw[-] (-2.75, -0.25) -- (-1, 1.5) -- (1.75, -1.25);
	\draw[-] (-2.25, -0.75) -- (0, 1.5) -- (2.25, -0.75);
	\draw[-,line width=0.3mm] (-1.75, -1.25) -- (1, 1.5) -- (2.75, -0.25) -- (0,-3) -- (-1.75, -1.25);
		
	\node (0) at (0, 2){$0$};
	\node (1) at (0, -1){$0$};
	\node (talagrand-cube) at (1.15, -1.25){$\bT_{\ell}$};

	
	\fill[color=purple!20!white!80](9, 1) -- (7.5, 2.5) -- (9, 4) -- (10.5, 2.5);
	\node (heads) at (5.5, 2.5){\includegraphics[width=1.1cm]{images/quarter-front}};
	\node (bh) at (5.5, 1.5){If $\bb_{\ell} = 1$:};
	\draw[] (10, 2) -- (8, 2);
	\draw[] (10, 3) -- (8, 3);
	
	\node() at (9, 3.4){$0$};
	\node() at (9, 2.5){$0$};
	\node() at (9, 1.6){$0$};
	
	\draw[-] (9, 1) -- (7.5, 2.5) -- (9, 4) -- (10.5, 2.5) -- (9, 1);
	
	
	\fill[color=purple!20!white!80] (9, -3.5) -- (7.5, -2) -- (9, -0.5) -- (10.5, -2) -- (9, -3.5);
	\node (tails) at (5.5, -2){\includegraphics[width=1.1cm]{images/jersey-quarter-back}};
	\node (bt) at (5.5, -3){If $\bb_{\ell} = 0$:};
	
	\fill[color=purple!20!blue!40!white!70] (10, -2.5) -- (9, -3.5) -- (8, -2.5);
	\fill[color=purple!20!blue!40!white!70] (10, -1.5) -- (9, -0.5) -- (8, -1.5);

	\draw[] (10, -2.5) -- (8, -2.5);
	\draw[] (10, -1.5) -- (8, -1.5);
	
		\draw[-] (9, -3.5) -- (7.5, -2) -- (9, -0.5) -- (10.5, -2) -- (9, -3.5);
	
	\node() at (9, -1.1){$1$};
	\node() at (9, -2){$0$};
	\node() at (9, -2.9){$1$};

	
	\draw[-latex,dashed,color=gray] (2.3, 0.1) -- (4.5, 2.25);
	\draw[-latex,dashed,color=gray] (2.3, 0.1) -- (4.5, -2.25);
	\node () at (2.26, 0.1) [circle,fill,inner sep=1pt]{};
	\node () at (1.9, 0.1) {$\overline{x}_{\bC}$};
 	
	\node (control) at (0, -4.75){$\zo^{\bC} \equiv \zo^m$};
	\node (action) at (9, -4.75){$\zo^{\bA} \equiv \zo^{a}$};

\end{tikzpicture}};

\draw[-latex,dashed,color=gray] (1,0) to[out=90,in=-90] (-3.5,3.25);
\draw[-latex,dashed,color=gray] (-1,0) to[out=-180,in=90] (-3.5,-2.75);

\node (A) at (-1,0) [circle,fill,inner sep=1.5pt,label=left:{\small $10$}]{};
\node (B) at (1,0) [circle,fill,inner sep=1.5pt,label=right:{\small $01$}]{};
\node (C) at (0,1) [circle,draw=black,fill=purple!20!white!80,inner sep=1.5pt,label=above:{\small $11$}]{};
\node (D) at (0,-1) [circle,draw=black,fill=purple!20!white!80,inner sep=1.5pt,label=below:{\small $00$}]{};

\node () at (3.5, 0) {$(y_1,y_2)\equiv\zo^2$};

\draw (A) -- (C) -- (B) -- (D) -- (A);

\end{tikzpicture}

\
	
\caption{A draw of $\fyes\sim\Dyes$. All our hypercubes adopt the convention that the bottom-most point is $(0,\ldots,0)$ and the topmost point is $(1,\ldots,1)$, and horizontal lines denote Hamming levels. Given an input $(x, y_1, y_2) \in \zo^{n}\times\zo^2$ we follow the arrows starting with $\zo^2$ in the center. The cross-hatched region in the control cube $\zo^{\bC}$ corresponds to inputs satisfying a unique Talagrand DNF term $\bT_\ell$. The pink regions correspond to $0$ assignments and blue regions to $1$ assignments.} 
\label{fig:int-yes}

\end{figure}


\begin{figure} \label{fig:no-intersecting}
\centering

\begin{tikzpicture}

\node (potato1) at (0,5.5) {\begin{tikzpicture}[scale=0.9]

	
	\fill[color=purple!20!white!80] (0, -3) -- (-3, 0) -- (0, 3) -- (3, 0);
	\fill[color=white] (-2.75, 0.25) -- (-1, -1.5) -- (1.75, 1.25) -- (0, 3);\fill[pattern=crosshatch,pattern color=black, opacity=0.175] (-2.75, 0.25) -- (-1, -1.5) -- (1.75, 1.25) -- (0, 3);
	\fill[color=white] (-2.25, 0.75) -- (0, -1.5) -- (2.25, 0.75) -- (0, 3);\fill[pattern=crosshatch,pattern color=black, opacity=0.175] (-2.25, 0.75) -- (0, -1.5) -- (2.25, 0.75) -- (0, 3);
	\fill[color=white] (-1.75, 1.25) -- (1, -1.5) -- (2.75, 0.25) -- (0,3) -- (-1.75, 1.25);\fill[pattern=crosshatch,pattern color=black, opacity=0.175] (-1.75, 1.25) -- (1, -1.5) -- (2.75, 0.25) -- (0,3) -- (-1.75, 1.25);
	
	\fill[color=purple!20!white!80] (-2.25, 0.75) -- (-0.5, -1) -- (0, -0.5) -- (0.5, -1) -- (2.25, 0.75) -- (0, 3);
	
	\draw[-] (0, -3) -- (-3, 0) -- (0, 3) -- (3, 0) -- (0, -3);

	\draw[-] (-2.75, 0.25) -- (-1, -1.5) -- (1.75, 1.25);
	\draw[-] (-2.25, 0.75) -- (0, -1.5) -- (2.25, 0.75);
	\draw[-,line width=0.3mm] (-1.75, 1.25) -- (1, -1.5) -- (2.75, 0.25) -- (0,3) -- (-1.75, 1.25);
		
	\node (0) at (0, -2){$0$};
	\node (1) at (0, 1){$0$};
	\node (talagrand-cube) at (1.15, -0.85){$\bT_{\ell}$};

	
	\fill[color=purple!20!white!80] (9, 1) -- (7.5, 2.5) -- (9, 4) -- (10.5, 2.5);
	\node (heads) at (5.5, 2.5){\includegraphics[width=1.1cm]{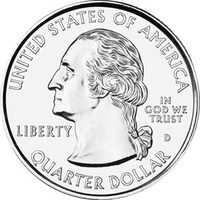}};
	\node (bh) at (5.5, 1.5){If $\bb_{\ell} = 1$:};
	\fill[color=purple!20!blue!40!white!70] (10, 2) -- (9, 1) -- (8, 2);
	\draw[] (10, 2) -- (8, 2);
	\draw[] (10, 3) -- (8, 3);
	
	\node() at (9, 3.4){$0$};
	\node() at (9, 2.5){$0$};
	\node() at (9, 1.6){$1$};
	
	\draw[-] (9, 1) -- (7.5, 2.5) -- (9, 4) -- (10.5, 2.5) -- (9, 1);
	
	\fill[color=purple!20!white!80] (9, -3.5) -- (7.5, -2) -- (9, -0.5) -- (10.5, -2);
	\node (tails) at (5.5, -2){\includegraphics[width=1.1cm]{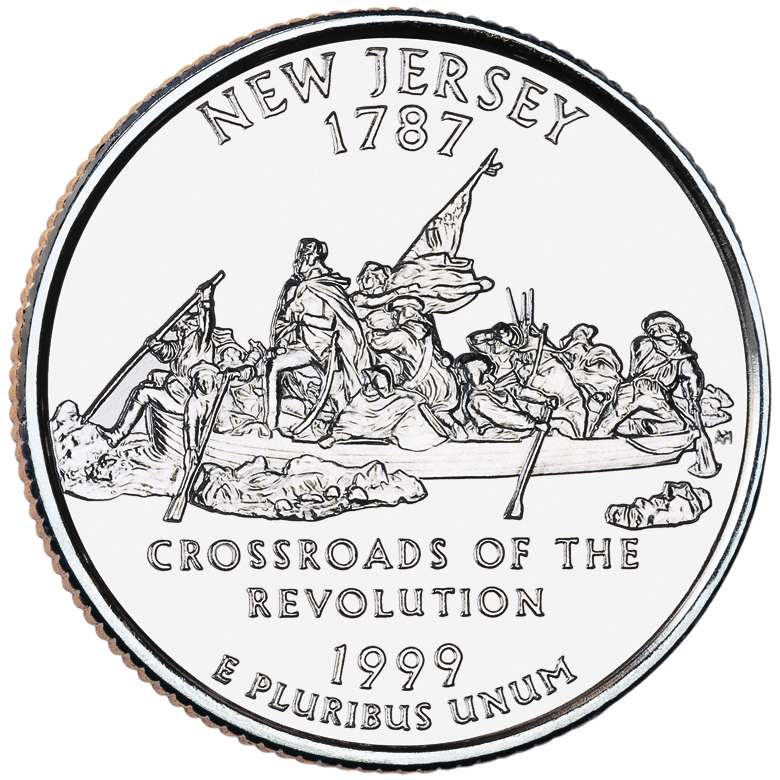}};
	\node (bt) at (5.5, -3){If $\bb_{\ell} = 0$:};
	
	\fill[color=purple!20!blue!40!white!70] (10,-1.5) -- (9, -0.5) -- (8, -1.5);
	\draw[-] (9, -3.5) -- (7.5, -2) -- (9, -0.5) -- (10.5, -2) -- (9, -3.5);	
	\draw[] (10, -2.5) -- (8, -2.5);
	\draw[] (10, -1.5) -- (8, -1.5);
	
	\node() at (9, -1.1){$1$};
	\node() at (9, -2){$0$};
	\node() at (9, -2.9){$0$};

	
	\draw[-latex,dashed,color=gray] (2.35, 0.1) -- (4.5, 2.25);
	\draw[-latex,dashed,color=gray] (2.35, 0.1) -- (4.5, -2.25);
	\node () at (2.35, 0.1) [circle,fill,inner sep=1pt]{};
		\node () at (2, 0.1) {${x}_{\bC}$};
	
	\node (control) at (0, -4.75){\phantom{$\zo^{\bC} \equiv \zo^m$}};
	\node (action) at (9, -4.75){\phantom{$\zo^{\bA} \equiv \zo^{a}$}};

\end{tikzpicture}};

\node (potato2) at (0,-6) {\begin{tikzpicture}[scale=0.9]

	
	\fill[color=purple!20!white!80] (0,-3) -- (-3,0) -- (0,3) -- (3,0);
	\fill[color=white] (-2.75, -0.25) -- (-1, 1.5) -- (1.75, -1.25) -- (0, -3);\fill[pattern=crosshatch,pattern color=black, opacity=0.175] (-2.75, -0.25) -- (-1, 1.5) -- (1.75, -1.25) -- (0, -3);
	\fill[color=white] (-2.25, -0.75) -- (0, 1.5) -- (2.25, -0.75) -- (0, -3);\fill[pattern=crosshatch,pattern color=black, opacity=0.175] (-2.25, -0.75) -- (0, 1.5) -- (2.25, -0.75) -- (0, -3);
	\fill[color=white] (-1.75, -1.25) -- (1, 1.5) -- (2.75, -0.25) -- (0,-3) -- (-1.75, -1.25);\fill[pattern=crosshatch,pattern color=black, opacity=0.175] (-1.75, -1.25) -- (1, 1.5) -- (2.75, -0.25) -- (0,-3) -- (-1.75, -1.25);
	
	\fill[color=purple!20!white!80] (-2.25, -0.75) -- (-0.5, 1) -- (0, 0.5) -- (0.5, 1) -- (2.25, -0.75) -- (0, -3);
	
	\draw[-] (0, -3) -- (-3, 0) -- (0, 3) -- (3, 0) -- (0, -3);

	\draw[-] (-2.75, -0.25) -- (-1, 1.5) -- (1.75, -1.25);
	\draw[-] (-2.25, -0.75) -- (0, 1.5) -- (2.25, -0.75);
	\draw[-,line width=0.3mm] (-1.75, -1.25) -- (1, 1.5) -- (2.75, -0.25) -- (0,-3) -- (-1.75, -1.25);
		
	\node (0) at (0, 2){$0$};
	\node (1) at (0, -1){$0$};
	\node (talagrand-cube) at (1.15, -1.25){$\bT_{\ell}$};

	
	\fill[color=purple!20!white!80] (9, 1) -- (7.5, 2.5) -- (9, 4) -- (10.5, 2.5);
	\node (heads) at (5.5, 2.5){\includegraphics[width=1.1cm]{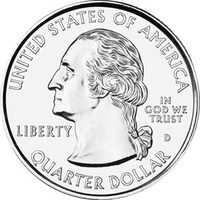}};
	\node (bh) at (5.5, 1.5){If $\bb_{\ell} = 1$:};
	\fill[color=purple!20!blue!40!white!70] (10, 2) -- (9, 1) -- (8, 2);
	\draw[] (10, 2) -- (8, 2);
	\draw[] (10, 3) -- (8, 3);
	
	\node() at (9, 3.4){$0$};
	\node() at (9, 2.5){$0$};
	\node() at (9, 1.6){$1$};
	
	\draw[-] (9, 1) -- (7.5, 2.5) -- (9, 4) -- (10.5, 2.5) -- (9, 1);
	
	\fill[color=purple!20!white!80] (9, -3.5) -- (7.5, -2) -- (9, -0.5) -- (10.5, -2);
	\node (tails) at (5.5, -2){\includegraphics[width=1.1cm]{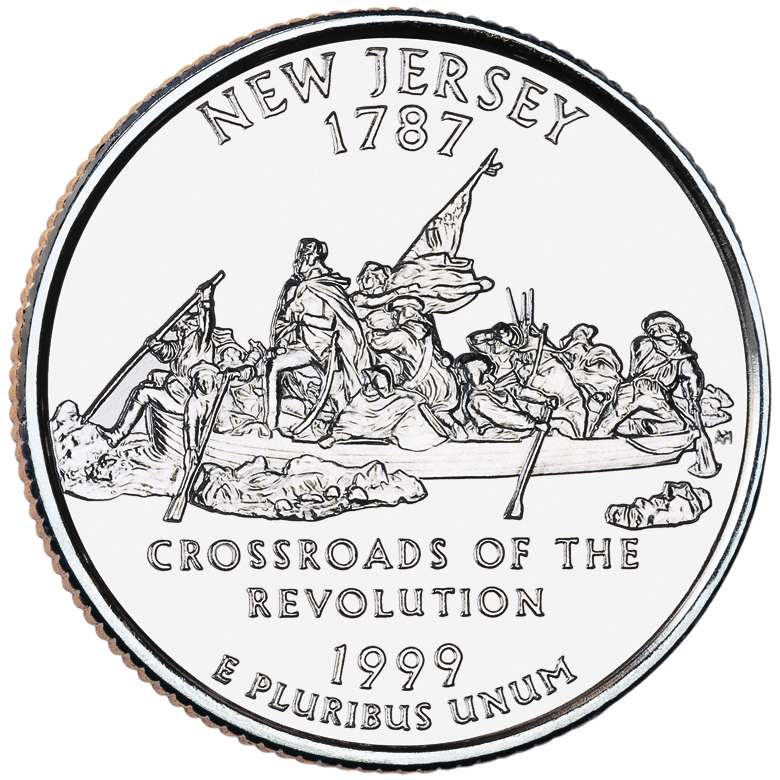}};
	\node (bt) at (5.5, -3){If $\bb_{\ell} = 0$:};
	
	\fill[color=purple!20!blue!40!white!70] (10,-1.5) -- (9, -0.5) -- (8, -1.5);
	\draw[-] (9, -3.5) -- (7.5, -2) -- (9, -0.5) -- (10.5, -2) -- (9, -3.5);	
	\draw[] (10, -2.5) -- (8, -2.5);
	\draw[] (10, -1.5) -- (8, -1.5);
	
	\node() at (9, -1.1){$1$};
	\node() at (9, -2){$0$};
	\node() at (9, -2.9){$0$};

	
	\draw[-latex,dashed,color=gray] (2.3, 0.1) -- (4.5, 2.25);
	\draw[-latex,dashed,color=gray] (2.3, 0.1) -- (4.5, -2.25);
	\node () at (2.26, 0.1) [circle,fill,inner sep=1pt]{};
	\node () at (1.9, 0.1) {$\overline{x}_{\bC}$};
 	
	\node (control) at (0, -4.75){$\zo^{\bC} \equiv \zo^m$};
	\node (action) at (9, -4.75){$\zo^{\bA} \equiv \zo^{a}$};

\end{tikzpicture}};

\draw[-latex,dashed,color=gray] (B) to[out=90,in=-90] (-3.5,3.25);
\draw[-latex,dashed,color=gray] (A) to[out=-180,in=90] (-3.5,-2.75);

\node (A) at (-1,0) [circle,fill,inner sep=1.5pt,label=left:{\small $10$}]{};
\node (B) at (1,0) [circle,fill,inner sep=1.5pt,label=right:{\small $01$}]{};
\node (C) at (0,1) [circle,draw=black,fill=purple!20!white!80,inner sep=1.5pt,label=above:{\small $11$}]{};
\node (D) at (0,-1) [circle,draw=black,fill=purple!20!white!80,inner sep=1.5pt,label=below:{\small $00$}]{};

\node () at (3.5, 0) {$(y_1,y_2)\equiv\zo^2$};

\draw (A) -- (C) -- (B) -- (D) -- (A);

\end{tikzpicture}

\
	
\caption{A draw of $\fno\sim\Dno$. Our conventions are as in \Cref{fig:int-yes}.}
\label{fig:int-no}

\end{figure}

\subsection{Distance to Intersectingness}
In this section, we will analyze the distance to intersectingness of functions in the support of $\Dy$ and $\Dn$. For conciseness, we write $m$ to denote $n-a$ throughout this subsection.

\begin{lemma}\label{lemma: Dyes is intersecting}
	Every function $f_{\mathrm{yes}}$ in the support of $\Dy$ is intersecting.
\end{lemma}
\begin{proof}
	Consider any function $f=f_{\mathrm{yes}}$ in the support of $\Dy$ and any two points $u,v\in\{0,1\}^{n+2}$ such that $f(u)=f(v)=1$. We will show that there exists $i\in[n+2]$ such that $u_i=v_i=1$.
	
	Since $f(u)=f(v)=1$, we have $u_{n+1}\neq u_{n+2}$ and $v_{n+1}\neq v_{n+2}$. If the last two bits of $u$ are the same as the last two bits of $v$ then we are done. So without loss of generality assume that $u=(x,0,1)$ and $v=(y,1,0)$ where $x,y\in\{0,1\}^n$. Then by definition of $f_{\mathrm{yes}}$ we must have that $|S_T(x_C)|=1$ and $|S_T(\overline{y}_C)|=1$.
	
	Let $S_T(x_C)=\{\ell\}$. We will show that $S_T(\overline{y}_C)\neq\{\ell\}$, which implies that $y_i=1$ for some $i\in T_{\ell}$, and hence  $x_i=y_i=1$ for that value of $i$. We show the contrapositive. If  $S_T(\overline{y}_C)=\{\ell\}$ as well, then we have $f(x,0,1)=g^{(+,b_{\ell})}(x_A)$ and $f(y,1,0)=g^{(+,\overline{b_{\ell}})}(y_A)$. But since $f(x,0,1)=1$, we have $b_{\ell}=1$. Since $g^{(+,0)}$ is constantly zero, we have $f(y,1,0)=g^{(+,\overline{b_{\ell}})}(y_A)=0$, leading to a contradiction.
	
	This finishes the proof.
\end{proof}

We next show that $\fno\sim \Dn$ is $\Omega(\eps)$-far from intersecting with high probability.

\begin{lemma}\label{lem:hehe2}
	With probability at least 0.01, 
	$\fno\sim\Dno$ is $\Omega(\eps)$-far from intersecting.
\end{lemma}
\begin{proof}	

We first introduce the event $\GoodTalagrand(T)$, which states that there exists an $\Omega(\eps)$-fraction of points $x \in \{0,1\}^{m}$ such that $|S_T(x)|=1$ (recall that $m=n-a$). 
Formally, let $t\cdot \eps$ be the number of $x$ in $\{0,1\}^m$ such that $|x| \in [m/2, m/2 + 0.05\eps\sqrt{m}]$ and let $\GoodTalagrand(T)$ be the event that $\mathbb{E}_{|\bm{x}|\in[m/2, m/2 + 0.05\eps\sqrt{m}]}\left[\textbf{1}_{\{|S_T(x)|=1\}}\right]\geq t\cdot \eps/100$. We will show that $$\prob_{\bm{T}\sim\Tal(m,\eps)}[\GoodTalagrand(\bm{T})]>0.02.$$

Let $p$ denote $\prob_{\bm{T}\sim\Tal(m,\eps)}[\GoodTalagrand(\bm{T})]$. 
Recall that in \Cref{prop:talagrand-unique-property}, we have shown that $\Prx_{\bm{T}}\sbra{|S_{\bm{T}}(x)| = 1} \geq 0.03$  for any $x\in\{0,1\}^{m}$ with $|x|\in [m/2, m/2 + 0.05\eps\sqrt{m}]$. So we have $$p\cdot t\eps+(1-p)\cdot t\eps/100\geq t\eps\Prx_{\bm{T}}\sbra{|S_{\bm{T}}(x)| = 1} \geq 0.03t\eps,$$ 
which implies that $p\geq 2/99> 0.02.$

Fix an arbitrary $T$ in the support of $\Tal(m,\eps)$ such that $\GoodTalagrand(T)$ happens. Note that for each $x$ such that $S_T(x)=\{\ell\}$ for some $\ell\in[L]$, $\mathbb{E}[\bm{b}_{\ell}]=1/2.$ So by  linearity of expectation and Markov's inequality, we know that with probability at least $99\%$, there is an $\Omega(\eps)$-fraction of points $x$ such that $S_T(x)=\{\ell\}$ and $b_{\ell}=1$. We write $\ManyOne$ to denote this event.

Assuming that both $\ManyOne$ and $\GoodTalagrand$ happen (whose probability is at least $1\%$), below we will show that $f_{\mathrm{no}}$ is $\Omega(\eps)$-far from intersecting, which will finish the proof.

For each $x_C\in\{0,1\}^m=\{0,1\}^{n-a}$ with $S_T(x_C)=\{\ell\}$ and $b_{\ell}=1$, we show that there is a set of violating pairs of size at least $\Omega(2^a)$. In particular, we show that there are at least $\Omega(2^a)$ pairs $(x^i_A,y^i_A)$ such that $f(x_C,x^i_A,0,1)=f(\overline{x}_C,y^i_A,1,0)=1$ but $x^i_A\cap y^i_A=\varnothing$. Then this lemma will follow from \Cref{lem:dist-to-intersecting}.

To this end, we include a folklore claim and its proof for completeness.
\begin{claim}\label{claim: k regular graph bipartite matching}
	For integer $0\leq w\leq a$, let $P_w$ denote the set of points in $\{0,1\}^a$ with Hamming weight $w$, i.e.~$P_w = \{x\in\{0,1\}^a \mid |x|=w\}$. Then for any $0\leq w< a/2$, the bipartite graph $(P_w,P_{a-w})$ with the poset relations as edges has a perfect matching.
\end{claim}
\begin{proof}

	The key point is that $(P_w,P_{a-w})$ is a $k$-regular bipartite graph, where $k={a-w \choose w}$.
	
	 We apply Hall's theorem to show this claim. Consider any subset $S\subseteq P_{a-w}$. The number of edges associated with $S$ is exactly $k|S|$. Let $\calN(S)$ be the neighborhood of $S$. Then we know the number of edges associated with $\calN(S)$ is exactly $k|\calN(S)|$, and these edges include the $k|S|$ edges above. So we have $k|S|\leq k|\calN(S)|$, which means $|S|\leq |\calN(S)|$.
\end{proof}

Fix any point $x_C\in\{0,1\}^m=\{0,1\}^{n-a}$ such that $S_T(x_C)=\{\ell\}$ and $b_{\ell}=1$ for some $\ell\in[L]$. Consider the action cube $\{0,1\}^a$. For any Hamming weight $w<\frac{a}{2}-\sqrt{a}$, fix a perfect matching $\{(x^i_A,\overline{y}^i_A)\}$ of $(P_w,P_{a-w})$. Then we note that $f(x_C,x^i_A,0,1)=1$, whereas $f(\overline{x}_C,y^i_A,1,0)=1$. The number of points with Hamming weight less than $\frac{a}{2}-\sqrt{a}$ is $\Omega(2^a)$, so there are $\Omega(2^a)$ violating pairs.

Then the lemma follows from the fact there are $\Omega(\eps\cdot 2^{n-a})$ points $x_C\in\{0,1\}^{n-a}$ such that $S_T(x_C)=\{\ell\}$ and $b_{\ell}=1$ for some $\ell\in[L]$ (since we assumed that both $\ManyOne$ and $\GoodTalagrand$ happen). So overall we have $\Omega(\eps)\cdot 2^n$ violating pairs, which means the function $f_{\mathrm{no}}$ is $\Omega(\eps)$-far from intersecting.
\end{proof}

\subsection{Indistinguishability of $\Dy$ and $\Dn$}\label{sec:hehe1}

In this section we establish the indistinguishability of the distributions $\Dy$ and $\Dn$. Specifically, for any nonadaptive deterministic algorithm $\calA$ with query complexity $q=2^{0.1n^{1/4}/\sqrt{\eps}}$, we show that
\begin{equation} \label{eq:goal}
\Prx_{\fyes\sim \Dy}[\calA \text{ accepts }\fyes]\leq \Prx_{\fno\sim \Dn}[\calA \text{ accepts }\fno]+ o_n(1).
\end{equation}
Our arguments closely follow the approach for proving indistinguishability that was used in \cite{CDLNS23}.

We begin with some simplifying assumptions: for any point $u\in\{0,1\}^{n+2}$ that is queried by the algorithm $\calA$ we assume that $u_{n+1}\neq u_{n+2}$ (since otherwise the answer to the query must be 0), and we assume that for each point $u \in \zo^{n+2}$ that is queried by ${\cal A}$ the point $\overline{u}$ is also queried as well (since this only affects the query complexity by at most a factor of two). So the set of $q$ query points of $\calA$ can be  characterized by a set $Q_{\calA}\coloneqq \{x^1,\cdots,x^q\}\subseteq \{0,1\}^n$, where both $(x^i,0,1)$ and $(\overline{x}^i,1,0)$ are queried for each $i\in[q]$.

A crucial step of the argument is that the only way for $\calA$ to distinguish $\Dy$ and $\Dn$ is to query two points $x^i,x^j$ with $S_T(x^i_C)=S_T(x^j_C)=\{\ell\}$ for some $\ell \in [L]$ such that one is in the top region and the other is in the bottom region of the action cube, namely $|x^i_A|>\frac{a}{2}+\sqrt{a}$ and $|x^j_A|<\frac{a}{2}-\sqrt{a}$. We let $\Bad$ denote this event (that $Q_{\calA}$ contains two points $x^i,x^j$ satisfying the above conditions).

Formally, let us write $\calA(f)$ to denote the sequence of $q$ answers to the queries made by $\calA$ to $f$.  We write $\mathrm{view}_{\calA}(\Dy)$ (respectively $\mathrm{view}_{\calA}(\Dn)$) to be the distribution of $\calA(f)$ for $\bm{f}\sim \Dy$ (respectively $\bm{f} \sim \Dn)$. The following claim asserts that conditioned on $\Bad$ not happening, the distributions $\mathrm{view}_{\calA}(\Dy|_{\overline{\Bad}})$ and $\mathrm{view}_{\calA}(\Dn|_{\overline{\Bad}})$ are identical.

\begin{lemma}\label{lemma: Bad is the only different place between Dy and Dn}
	$\mathrm{view}_{\calA}(\Dy|_{\overline{\Bad}})=\mathrm{view}_{\calA}(\Dn|_{\overline{\Bad}}).$
\end{lemma}
\begin{proof}
	The distributions of the partition of $[n]$ into control variables $\bm{C}$ and action variables $\bm{A}$ are identical for $\Dy$ and $\Dn$. So fix an arbitrary partition $C$ and $A$. As the distribution of the Talagrand DNF $\bm{T}\sim\Tal(m,\eps)$ is also identical, we fix an arbitrary $T$.
	
	We divide the points $Q_{\calA}$ into disjoint groups according to $x_C$. More precisely, for every $\ell\in [L]$, let $Q_{\calA}(\ell)=\{x^i\mid S_T(x^i_C)=\{\ell\}\}$. The points outside $\bigcup_{\ell\in[L]} Q_{\calA}(\ell)$ are not important as $\bm{f}$ will be identically 0 for both $\Dy$ and $\Dn$.

	Let $\bm{f}_{\ell}(x)$ denote the function $\bm{f}(x,0,1)$ restricted to points in $Q_{\calA}(\ell)$, and let $\bm{f}'_{\ell}(x)$ similarly denote the function $\bm{f}(\overline{x},1,0)$ restricted to inputs $x\in Q_{\calA}(\ell)$. Note that for a fixed $\ell\in[L]$, the functions $\bm{f}_{\ell}(x)$ and $\bm{f}'_{\ell}(x)$ only depend on the random bit $\bm{b}_{\ell}$. As a result, the distributions of functions $\bm{f}_{\ell}(x)$ and $\bm{f}'_{\ell}(x)$ for different $\ell$ are independent.
	
	So fix an arbitrary $\ell\in[L]$. The condition that $\Bad$ does not happen implies that \textit{either} $|x_A|>a/2+\sqrt{a}$ for all $x\in Q_{\calA}(\ell)$ \textit{or} $|x_A|<a/2-\sqrt{a}$ for all $x\in Q_{\calA}(\ell)$, which holds for both $\Dy$ and $\Dn$. So we have $\bm{f}'_{\ell}(x)=1-\bm{f}_{\ell}(x)$ for all $x\in Q_{\calA}(\ell)$, which also holds for both $\Dy$ and $\Dn$.
	
	Finally, noticing that the distribution of $\bm{f}_{\ell}(x)$ is simply a uniform random bit $\bm{b}_{\ell}$ for both $\Dy$ and $\Dn$, this finishes the proof.
\end{proof}

Next, we show that the probability that $\Bad$ happens is small
(recall that $q=2^{0.1n^{1/4}/\sqrt{\eps}}$):

\begin{lemma}\label{lemma: Bad is unlikely}
	For any set of points $Q_{\calA}=\{x^1,\cdots,x^q\}\subseteq \{0,1\}^n$, $\prob[\Bad]=o_n(1)$.
\end{lemma}
\begin{proof}
	Fix any two points $x,y\in\{0,1\}^n$. We will upper bound the probability that $S_{\bT}(x_{\bm{C}})=S_{\bT}(y_{\bm{C}})=\{\ell\}$ for some $\ell\in[L]$ and $|x_{\bm{A}}|<\frac{a}{2}-\sqrt{a}$ and $|y_{\bm{A}}|>\frac{a}{2}+\sqrt{a}$. Call this specific event $\Bad_{xy}$.
	
	Let $I_{01}$ be the set of indices $i$ such that $x_i=0$ and $y_i=1$. On the one hand, to have $\Bad_{xy}$ happen, we must have that
	\begin{equation} \label{eq:diamond}
		|I_{01}\cap \bm{A}|\geq 2\sqrt{a}. \tag{$\diamond$}
	\end{equation}
	
	On the other hand, to have  $S_{\bT}(x_{\bm{C}})=S_{\bT}(y_{\bm{C}})=\{\ell\}$, we must have that 
	\begin{equation}
		\text{There exists an } \ell \in [L]\text{ such that }S_{\bT}(x)=S_{\bT}(y)=\{\ell\}. \tag{$\star$}
	\end{equation}

So we have $\prob[\Bad_{xy}]\leq \min(\prob[\diamond],\prob[\star])$; we will show that
\[
\min(\prob[\diamond],\prob[\star])
\leq 2^{-0.05n^{1/4}/\sqrt{\eps}}.
\]
	
	Let $t=|I_{01}|$. Then by the random choice of the coordinates defining the action cube $\bA$, we have
\begin{align*}
    \prob[\diamond] & \leq \prob\left [\mathrm{Bin}\left (a,\frac{t}{n-a}\right )\geq 2\sqrt{a}\right]\leq {a\choose 2\sqrt{a}}\cdot \left(\frac{t}{n-a}\right)^{2\sqrt{a}}\\
    & \leq \left(\frac{ea}{2\sqrt{a}}\right)^{2\sqrt{a}}\cdot \left(\frac{t}{n-a}\right)^{2\sqrt{a}}\leq \left(\frac{et\sqrt{a}}{2(n-a)}\right)^{2\sqrt{a}}\leq \left(\frac{et\sqrt{a}}{2(1-\frac{1}{c_0})n}\right)^{2\sqrt{a}}.
\end{align*}
To bound $\prob[\star]$, we use
\begin{align*}
\prob[\star]&=\prob[S_{\bT}(x)=S_{\bT}(y) \ \& \ \exists \ell \in [L] \text{~such that~}S_{\bT}(y)=\{\ell\}]\\
&\leq \prob[S_{\bT}(x)=S_{\bT}(y) \ \mid \ \exists \ell \in [L] \text{~such that~}S_{\bT}(y)=\{\ell\}]\\
&\leq \max_{\ell \in [L]} \prob[S_{\bT}(x)=S_{\bT}(y) \ \mid \ S_{\bT}(y)=\{\ell\}]\\
&\leq \left(1-\frac{t}{n-a}\right)^{\sqrt{n-a}/\eps}\leq e^{-t/(\eps\sqrt{n-a})}\leq e^{-t/(\eps\sqrt{n})},
\end{align*}
where the last line above is by the definition of the random process $\bT \sim \Tal(n-a,\eps)$.

When $t\leq \frac{1}{4}n^{3/4}/\sqrt{\eps}$, we have $\prob[\diamond]\leq 2^{-n^{1/4}/\sqrt{\eps}}.$ When $t\geq \frac{1}{4}n^{3/4}/\sqrt{\eps}$, we have $\prob[\star]\leq 2^{-0.25n^{1/4}/\sqrt{\eps}}.$

	So overall we have $$\prob[\Bad_{xy}]\leq \min(\prob[\diamond],\prob[\star])\leq 2^{-0.25n^{1/4}/\sqrt{\eps}}.$$
	
	By a union bound for all pairs of points of $Q_{\calA}$, we know that $$\prob[\Bad]\leq 2^{-0.25n^{1/4}/\sqrt{\eps}}\cdot \left(2^{0.1n^{1/4}/\sqrt{\eps}}\right)^2=o_n(1),$$
	and the lemma is proved.
\end{proof}

Now we are ready to prove \Cref{thm:main-lb-intersecting}.

\begin{proof}[Proof of \Cref{thm:main-lb-intersecting}]
Let $\calD=\frac{1}{2}\{\Dy+\Dn\}$. Then we have
\begin{align}
    \Prx_{\bm{f}\sim \calD}[\calA \text{ is correct on }\bm{f}]&= \frac{1}{2}\left(\Prx_{\fyes\sim \Dy}[\calA \text{ is correct on }\fyes]+\Prx_{\fno\sim \Dn}[\calA \text{ is correct on }\fno]\right)\nonumber\\
    &= \frac{1}{2}\left(\Prx_{\fyes\sim \Dy}[\calA \text{ accepts }\fyes]+\Prx_{\fno\sim \Dn}[\calA \text{ is correct on }\fno]\right)\label{line: 2}\\
    &\leq \frac{1}{2}\left(\Prx_{\fyes\sim \Dy}[\calA \text{ accepts }\fyes]+0.99+0.01\Prx_{\fno\sim \Dn}[\calA \text{ rejects }\fno]\right)\label{line: 3}\\
    &= \frac{1}{2}\left(\Prx_{\fyes\sim \Dy}[\calA \text{ accepts }\fyes]+1-0.01\Prx_{\fno\sim \Dn}[\calA \text{ accepts }\fno]\right)\nonumber\\
    &\leq \frac{199}{200}+\frac{1}{200}\left(\Prx_{\fyes\sim \Dy}[\calA \text{ accepts }\fyes]-\Prx_{\fno\sim \Dn}[\calA \text{ accepts }\fno]\right)\nonumber\\
    &= \frac{199}{200}+\frac{\Pr[\Bad]}{200}\left(\Prx_{\fyes\sim \Dy|_{\Bad}}[\calA \text{ accepts }\fyes]-\Prx_{\fno\sim \Dn|_{\Bad}}[\calA \text{ accepts }\fno]\right)\label{line: 6}\\
    &\leq \frac{199}{200}+\frac{\Pr[\Bad]}{200}\nonumber\\
    & \leq \frac{199}{200}+ o_n(1)\label{line: 8},
\end{align}
where \Cref{line: 2} is because of \Cref{lemma: Dyes is intersecting}, \Cref{line: 3} is because $\fno$ is not $\eps$-far from intersecting with probability at most 0.99 thanks to \Cref{lem:hehe2}, \Cref{line: 6} is from \Cref{lemma: Bad is the only different place between Dy and Dn}, and \Cref{line: 8} follows from \Cref{lemma: Bad is unlikely}. \Cref{thm:main-lb-intersecting} now follows from Yao's minimax principle (\Cref{thm:yao-minimax}).
\end{proof}

\subsection{A $2^{\Omega(\sqrt{n \log(1/\eps)})}$ Lower Bound for One-Sided Non-adaptive Testers of Intersectingness}
\label{sec:intersecting-one-sided}

In this section we prove~\Cref{thm:one-sided-lb-intersecting}, by giving a $2^{\Omega(\sqrt{n \log(1/\eps)})}$-query complexity lower bound against any \emph{non-adaptive} and \emph{one-sided} algorithm testing $\eps$-intersectingness. 
This almost matches the query complexity of our $n^{O(\sqrt{n\log(1/\eps)})}/\eps$-query one-sided non-adaptive algorithm even for constant $\eps$.

Since we are working against one-sided algorithms, it suffices for us to describe a distribution $\Dn$ over $f:\{0,1\}^{n+2}\rightarrow\{0,1\}$ of ``no''-functions (functions that are far from intersecting). Let $K=\sqrt{n\ln(1/\eps)}$. A draw from our $\Dn$ distribution is obtained as follows: first, we sample a subset $\bm{A}\subseteq [n]$ of size $a=n/100$ uniformly at random (looking ahead, 100 will be an important constant later in the proof). Then $\fno\sim\Dn$ is defined by letting $\fno(x,0,0)=\fno(x,1,1)=0$ for all $x\in\{0,1\}^n$, and

\begin{equation*}
\fno(x,0,1) = \fno(x,1,0) =
	\begin{cases}
       \ 0 & |x|\not\in [n/2-10K,n/2+10K]; \\[0.5ex]
       \ 0 & |x_{\bm{A}}|>n/200+K;\\[0.5ex]
       \ 0 & |x_{\bm{A}}|\in[n/200-K,n/200+K];\\[0.5ex]
       \ 1 & |x_{\bm{A}}|<n/200-K.
    \end{cases}
\end{equation*}
The constant ``10'' above will also be important vis-a-vis the ``100'' in the definition of the size of $\bm{A}$.

We first show that every $\fno\sim \Dn$ is $\eps^{O(1)}$-far from intersecting (observe that this suffices for our claimed lower bound, since the difference between $\eps$ and $\eps^{O(1)}$ is swallowed up by the log and the big-Omega):

\begin{lemma}
	Every $\fno$ in the support of 	$\Dn$ is $\eps^{O(1)}$-far from intersecting.
\end{lemma}
\begin{proof}
	Fix an arbitrary $A\subseteq [n]$ with size $a=n/100$, which  determines a function $f_{\mathrm{no}}$ in the support of $\Dn$. For the convenience of notations, we use $C\coloneqq [n]\setminus A$.
	
	We first recall \Cref{claim: k regular graph bipartite matching} on the Boolean cube $\{0,1\}^A$. For any $w\in [0,n/100]$, let $\calP_w\coloneqq \{x \in \zo^A]\mid |x|=w\}$. By the same argument as~\Cref{claim: k regular graph bipartite matching}, we know for any $0\leq w< n/200$, the bipartite graph $(P_{w},P_{n/100-w})$ with poset relations as edges has a perfect matching. 
	
	Next, we use the Chernoff bound (which upper bounds the lower tail of the Binomial distribution) and a ``reverse Chernoff bound'' (which \emph{lower} bounds the lower tail  of the Binomial distribution) to show that 
	$$|\{x\in\{0,1\}^A\mid |x|\in[n/200-5K,n/200-K]\}|\geq (\eps^{1800}-\eps^{5000})2^{a}=\Omega(\eps^{1800})\cdot 2^a.$$
	To this end, for $w\in[0,n/200]$, let $\calP_{\leq w}$ to denote $\{x\in\zo^A\mid |x|\leq w\}$. Then it suffices to show that 
	$$|\calP_{\leq n/200-5K}|\leq \eps^{5000}\cdot 2^a,$$ 
	which follows from the standard Chernoff bound, and
	$$|\calP_{\leq n/200-K}|\geq \eps^{1800}\cdot 2^a,$$
	which follows from the following ``reverse Chernoff bound:''
	\begin{lemma}[\cite{KY15}, Lemma 4]
		Let $X$ be the sum of $k$ independent 0/1 random variables. For any $K\in(0,pk/2]$ and $p\in[0,1/2]$ such that $K^2/(pk)\geq 3$, if each random variable is 1 with probability at most $p$, then
		$$\Pr[X\leq pk-K]\geq \exp(-9K^2/(pk)).$$
	\end{lemma}
	
	Next, consider any $x\in\zo^n$ such that $|x|\in[n/2-10K,n/2]$ and $|x_A|\in[n/200-5K,n/200-K]$. Let $y\in\zo^n$ be such that $y_C=x_C$ and $y_A$ is the matched point of $x_C$ in the perfect matching. Then we have $|y|\in[n/2-10K,n/2+10K]$ and $|y_A|\in [n/200+K,n/200+5K]$.
	
	Note that for any such pair $(x,y)$ we have $x\leq y$, $f(x,0,1)=1$ and $f(\overline{y},1,0)=1$, which serves as an $\textsf{I}$-violating pair. Since the edges in a perfect matching are vertex-disjoint, we have the number of $\textsf{I}$-violating pairs is at least the number of $x\in\zo^n$ such that $|x|\in[n/2-10K,n/2]$ and $|x_A|\in[n/200-5K,n/200-K]$.
	
	We have shown that $$|\{x\in\{0,1\}^A\mid |x|\in[n/200-5K,n/200-K]\}|=\Omega(\eps^{1800})\cdot 2^a.$$
	Note also that $$|\{x\in\{0,1\}^C\mid |x|\in[99n/200-5K,99n/200]\}|=\Omega(1)\cdot 2^{n-a}.$$
	
	This finishes the proof.
\end{proof}

Below we show that for any nonadaptive deterministic query algorithm $\calA$ with query complexity $q=2^{0.9\sqrt{n \log(1/\eps)}}$ the probability that $\calA$ succeeds in finding a violation of intersectingness is $o_n(1)$; this proves \Cref{thm:one-sided-lb-intersecting}. 
\begin{proof}[Proof of \Cref{thm:one-sided-lb-intersecting}]
We establish the following lemma, from which the theorem follows by a straightforward union bound:
	\begin{lemma}
		For any two points $x,y\in\{0,1\}^n$ such that $|x|,|y|\in [n/2-10K,n/2+10K]$ and $x\leq y$, $$\Prx_{\bm{A}}[|x\cap \bm{A}|<n/200-K \text{ and } |y\cap \bm{A}|>n/200+K]\leq 2^{-2K}.$$
	\end{lemma}
	\begin{proof}
		Let $I$ be the indices $i$ such that $x_i=0$ and $y_i=1$ and let $t=|I|$. Then we know $0\leq t\leq 20K$. On the other hand, in order for the event $|x\cap \bm{A}|<n/200-K \text{ and } |y\cap \bm{A}|>n/200+K$ to happen, the set $\bm{A}$ has to hit at least $2K$ many indices in $I$.
		
		So \begin{align*}
			&\quad\Prx_{\bm{A}}[|x\cap \bm{A}|<n/200-K \text{ and } |y\cap \bm{A}|>n/200+K] \\
			& \leq \Prx\left[\text{Bin}\left(n/100,\frac{20K}{0.99n}\right)\geq 2K\right]\leq {n/100\choose 2K}\cdot\left(\frac{20K}{0.99n}\right)^{2K}\\
			&\leq \left(\frac{en/100}{2K}\right)^{2K}\cdot\left(\frac{20K}{0.99n}\right)^{2K}\\
			&\leq \left(\frac{10e}{99}\right)^{2K}\leq 
			2^{-2K},
		\end{align*}
		completing the proof.
	\end{proof}
	By a union bound over all pairs of query strings where
	$q=2^{0.9K}=2^{0.9\sqrt{n \log(1/\eps)}}$,
	it follows that the probability that $\calA$ succeeds in finding a violation of intersectingness is $o_n(1)$. Since a one-sided tester must find such a violation in order to reject, this finishes the proof.
\end{proof}


\section{Lower bounds for Testing Union-Closed Families}
\label{sec:union}

In this section, we prove a $n^{\Omega(\log(1/\epsilon))}$-query lower bound against non-adaptive algorithms for testing union-closedness (with either one-sided or two-sided error). We describe the hard distributions in \Cref{subsec:uc-yes-no-dists} and then prove \Cref{thm:two-sided-lb-UC} in \Cref{subsec:uc-indistinguish}.

\subsection{The $\Dyes$ and $\Dno$ Distributions}
\label{subsec:uc-yes-no-dists}

Our construction of the hard distributions $\Dyes$ and $\Dno$ are inspired by the constructions for the lower bound against intersectingness testing in \Cref{sec:intersecting}; in particular, our hard functions will also comprise of a truncated Talagrand random DNF on a set of ``control bits'' $\bC$, and then a function tailored to the union-closedness property on a set of ``action bits'' $\bA$. We describe illustrate both $\Dyes$ and $\Dno$ in \Cref{fig:uc-lb}, and start by describing the $\Dyes$ distribution:


\begin{figure}
\centering

\begin{tikzpicture}[scale=0.9]

	
	\fill[color=purple!20!white!80] (0, -3) -- (-3, 0) -- (0, 3) -- (3, 0);
	\fill[color=white] (-2.75, 0.25) -- (-1, -1.5) -- (1.75, 1.25) -- (0, 3);\fill[pattern=crosshatch,pattern color=black, opacity=0.175] (-2.75, 0.25) -- (-1, -1.5) -- (1.75, 1.25) -- (0, 3);
	\fill[color=white] (-2.25, 0.75) -- (0, -1.5) -- (2.25, 0.75) -- (0, 3);\fill[pattern=crosshatch,pattern color=black, opacity=0.175] (-2.25, 0.75) -- (0, -1.5) -- (2.25, 0.75) -- (0, 3);
	\fill[color=white] (-1.75, 1.25) -- (1, -1.5) -- (2.75, 0.25) -- (0,3) -- (-1.75, 1.25);\fill[pattern=crosshatch,pattern color=black, opacity=0.175] (-1.75, 1.25) -- (1, -1.5) -- (2.75, 0.25) -- (0,3) -- (-1.75, 1.25);
	
	\fill[color=purple!20!blue!40!white!70] (-2.25, 0.75) -- (-0.5, -1) -- (0, -0.5) -- (0.5, -1) -- (2.25, 0.75) -- (0, 3);
	
	\draw[-] (0, -3) -- (-3, 0) -- (0, 3) -- (3, 0) -- (0, -3);

	\draw[-] (-2.75, 0.25) -- (-1, -1.5) -- (1.75, 1.25);
	\draw[-] (-2.25, 0.75) -- (0, -1.5) -- (2.25, 0.75);
	\draw[-,line width=0.3mm] (-1.75, 1.25) -- (1, -1.5) -- (2.75, 0.25) -- (0,3) -- (-1.75, 1.25);
		
	\node (0) at (0, -2){$0$};
	\node (1) at (0, 1){$1$};
	\node (talagrand-cube) at (1.15, -0.85){$\bT_{\ell}$};

	
	\fill[color=purple!20!white!80] (9, -1.5) -- (7.5, 0) -- (9, 1.5) -- (10.5, 0);
	\draw[-] (9, -1.5) -- (7.5, 0) -- (9, 1.5) -- (10.5, 0) -- (9, -1.5);
	
	\node[circle,fill,inner sep=1pt] (arr) at (9.25, .15){};
	\node (aaa) at (9.25, -0.1) {\small$\bs_\ell$};

	
	\draw[-latex,dashed,color=gray] (2.34, 0) -- (7, 0);
	\node () at (2.34, 0) [circle,fill,inner sep=1pt]{};
	\node (ex) at (2, 0.1) {$x_{\bC}$};
	\node (control) at (0, -4.75){$\zo^{\bC} \equiv \zo^m$};
	\node (action) at (9, -4.75){$\zo^{\bA} \equiv \zo^{a}$};

\end{tikzpicture}

\caption*{(a) A draw of $\fyes\sim\Dyes$}

\

\
	
\begin{tikzpicture}[scale=0.9]

	
	\fill[color=purple!20!white!80] (0, -3) -- (-3, 0) -- (0, 3) -- (3, 0);
	\fill[color=white] (-2.75, 0.25) -- (-1, -1.5) -- (1.75, 1.25) -- (0, 3);\fill[pattern=crosshatch,pattern color=black, opacity=0.175] (-2.75, 0.25) -- (-1, -1.5) -- (1.75, 1.25) -- (0, 3);
	\fill[color=white] (-2.25, 0.75) -- (0, -1.5) -- (2.25, 0.75) -- (0, 3);\fill[pattern=crosshatch,pattern color=black, opacity=0.175] (-2.25, 0.75) -- (0, -1.5) -- (2.25, 0.75) -- (0, 3);
	\fill[color=white] (-1.75, 1.25) -- (1, -1.5) -- (2.75, 0.25) -- (0,3) -- (-1.75, 1.25);\fill[pattern=crosshatch,pattern color=black, opacity=0.175] (-1.75, 1.25) -- (1, -1.5) -- (2.75, 0.25) -- (0,3) -- (-1.75, 1.25);

	\fill[color=purple!20!blue!40!white!70] (-2.25, 0.75) -- (-0.5, -1) -- (0, -0.5) -- (0.5, -1) -- (2.25, 0.75) -- (0, 3);
	
	\draw[-] (0, -3) -- (-3, 0) -- (0, 3) -- (3, 0) -- (0, -3);

	\draw[-] (-2.75, 0.25) -- (-1, -1.5) -- (1.75, 1.25);
	\draw[-] (-2.25, 0.75) -- (0, -1.5) -- (2.25, 0.75);
	\draw[-,line width=0.3mm] (-1.75, 1.25) -- (1, -1.5) -- (2.75, 0.25) -- (0,3) -- (-1.75, 1.25);
		
	\node (0) at (0, -2){$0$};
	\node (1) at (0, 1){$1$};
	\node (talagrand-cube) at (1.15, -0.85){$\bT_{\ell}$};

	
	
	\node (heads) at (5.5, 2.5){\includegraphics[width=1.1cm]{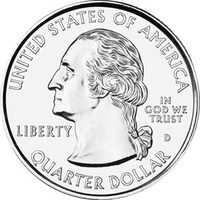}};
	\node (bh) at (5.5, 1.5){If $\bb_{\ell} = 1$:};

	\fill[color=purple!20!white!80] (9, 1) -- (7.5, 2.5) -- (9, 4) -- (10.5, 2.5);
	\draw[-] (9, 1) -- (7.5, 2.5) -- (9, 4) -- (10.5, 2.5) -- (9, 1);
	
	\node[circle,fill,inner sep=1pt] (arr) at (8.5, 2.75){};
	\node (ol-arr)[circle,fill,inner sep=1pt] at (9.5, 2.25){};
	
	\draw[] (arr) -- (9, 4) -- (ol-arr);
	
	\node (arrlab) at (8.5,2.5){\small$\br_{\ell}$};
	\node (ol-arrlab) at (9.5,2){\small$\overline{\br}_{\ell}$};
	

	
	\node (tails) at (5.5, -2){\includegraphics[width=1.1cm]{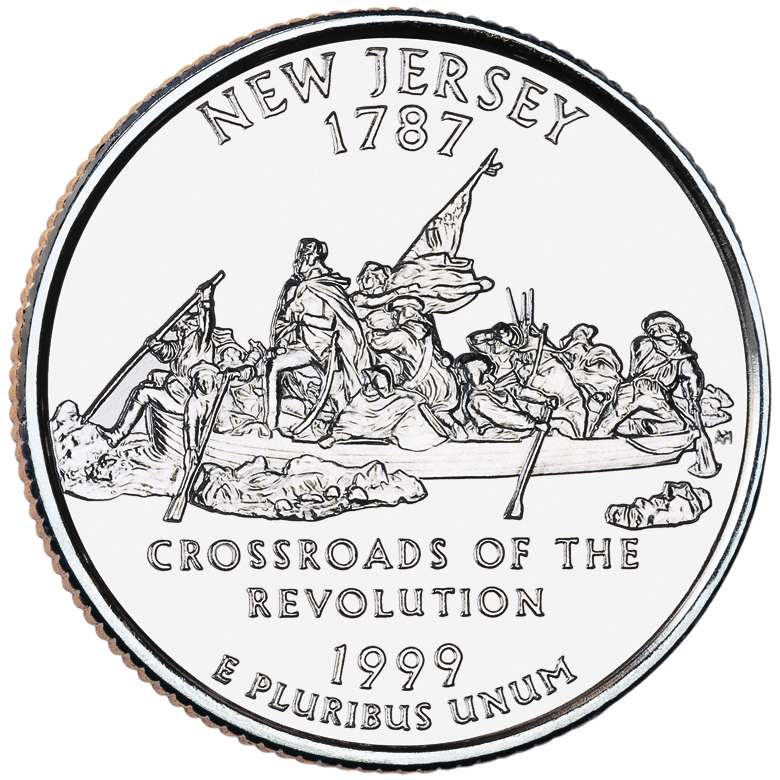}};
	\node (bt) at (5.5, -3){If $\bb_{\ell} = 0$:};
	
	\fill[color=purple!20!white!80] (9, -3.5) -- (7.5, -2) -- (9, -0.5) -- (10.5, -2);
	\node (o) at (9, -2) {$0$};
	\draw[-] (9, -3.5) -- (7.5, -2) -- (9, -0.5) -- (10.5, -2) -- (9, -3.5);

	
	\draw[-latex,dashed,color=gray] (2.35, 0.1) -- (4.5, 2.25);
	\draw[-latex,dashed,color=gray] (2.35, 0.1) -- (4.5, -2.25);

	\node (control) at (0, -4.75){$\zo^{\bC} \equiv \zo^m$};
	\node (action) at (9, -4.75){$\zo^{\bA} \equiv \zo^{a}$};

	\node () at (2.35, 0.1) [circle,fill,inner sep=1pt]{};
	\node (ex) at (2, 0.1) {$x_{\bC}$};

\end{tikzpicture}

\

\caption*{(b) A draw of $\fno\sim\Dno$}

\
	
\caption{An illustration of the yes- and no-distributions for the union-closedness lower bound. Our conventions are as in \Cref{fig:int-yes}. In (b), if $\bb_{\ell} = 1$ then as long as $\br_\ell \notin \{0^a,1^a\}$ the action cube $\zo^{\bA}$ will contain a single violation of union-closedness.}
	\label{fig:uc-lb}
\end{figure}

\begin{definition} \label{def:uc-yes}
	Given $\eps > 0$, a draw of a Boolean function $\fyes \isazofunc$ from the distribution $\Dyes := \Dyes(n,\eps)$ is obtained as follows:
	\begin{enumerate}
		\item Draw a random set of $a := \log(1/\eps)$ coordinates $\bA \sse[n]$, i.e. 
		\[\bA \sim {[n]\choose a}, \qquad\text{and set}\qquad \bC := [n]\setminus \bA.\]
		Let $c := |\bC| = n-a$.
		\item Let $L :=0.1\cdot2^{\sqrt{c}}$ and draw an $L$-term monotone Talagrand DNF $\bT\sim \Tal(c, 1)$ as defined in \Cref{def:talagrand} on $\zo^{\bC}$. 
		\item For each $\ell\in[L]$, independently draw a uniformly random $a$-bit string $\bs_\ell \in \zo^{\bA}$. 
		\item Output the function 
		\begin{equation*} \label{eq:uc-yes}
			\fyes(x_{\bC}, x_{\bA}) := 
			\begin{cases}
				1 & |S_{\bT}(x_{\bC})| \geq 2 \\
				\Indicator\cbra{x_{\bA} = \bs_\ell} & S_{\bT}(x_{\bC}) = \{\ell\}\\
				0 & |S_{\bT}(x_{\bC})| = 0
			\end{cases}
		\end{equation*}
		where $S_{\bT}$ is as defined in \Cref{notation:term-count}.
	\end{enumerate}
\end{definition}

We verify that functions drawn from $\Dyes$ are indeed union-closed:

\begin{claim} \label{claim:uc-yes-indeed-yes}
	Every function $f_\yes$ in the support of $\Dyes$ is union-closed.
\end{claim}

\begin{proof}
	Let $f_\yes$ be as in \Cref{def:uc-yes}. We will show the following:
	\[\text{If}~ f_\yes(x) = f_\yes(y) = 1,~\text{then}~f_\yes(x\cup y) = 1.\]
	For brevity we denote $z:= x \cup y$. Now, note that if $f_\yes(x) = 1$, then either 
	\begin{enumerate}
		\item[(a)] $|S_{\bT}(x_{\bC})| \geq 2$; or 
		\item[(b)] $S_{\bT}(x_{\bC}) = \{\ell\}$ for some $\ell\in[L]$ and $x_{\bA} = \bs_\ell$,
	\end{enumerate}
	with analogous statements holding for $y$. Note that if (a) holds for either $x$ or $y$, then $|S_{\bT}(z_{\bC})| \geq 2$ and so $f_\yes(z) = 1$. We therefore only have to consider the case when (b) holds for both $x$ and $y$. Indeed, suppose that 
	\[S_{\bT}(x_{\bC}) = \{\ell_1\} \qquad\text{and}\qquad S_{\bT}(y_{\bC}) = \{\ell_2\}.\]
	If $\ell_1\neq\ell_2$, then $\{\ell_1, \ell_2\} \sse S_{\bT}(z_{\bC})$ and so $f_\yes(z) = 1$. On the other hand, if $\ell_1 = \ell_2$ then either (i) $S_{\bT}(z_{\bC}) = \{\ell_1\}$, in which case as $\bs_{\ell_1} = x_{\bA} = y_{\bA} = z_{\bA}$ we will have $f_\yes(z) = 1$; or (ii) $|S_{\bT}(z_{\bC})|\geq 2$ and so $f_\yes(z) = 1$. It follows that $f_\yes$ is union closed.
\end{proof}

We now turn to a description of the $\Dno$ distribution. 

\begin{definition} \label{def:uc-no}
	Given $\eps > 0$, a draw of a Boolean function $\fno \isazofunc$ from the distribution $\Dno := \Dno(n,\eps)$ is obtained as follows:
	\begin{enumerate}
		\item Draw a random set of $a := \log(1/\eps)$ coordinates $\bA \sse[n]$, i.e. 
		\[\bA \sim {[n]\choose a}, \qquad\text{and set}\qquad \bC := [n]\setminus \bA.\]
		Let $c := |\bC| = n-a$.
		\item Let $L :=0.1\cdot2^{\sqrt{c}}$ and draw an $L$-term monotone Talagrand DNF $\bT\sim \Tal(c, 1)$ as defined in \Cref{def:talagrand} on $\zo^{\bC}$. 
		\item For each $\ell\in[L]$, independently draw a uniformly random $a$-bit string $\br_\ell \in \zo^{\bA}$ as well as a uniformly random bit $\bb_\ell \in \zo$. 
		\item Output the function 
		\begin{equation*} \label{eq:uc-yes}
			\fyes(x_{\bC}, x_{\bA}) := 
			\begin{cases}
				1 & |S_{\bT}(x_{\bC})| \geq 2 \\
				\bb_\ell \cdot\Indicator\cbra{x_{\bA} \in \{\br_\ell, \overline{\br}_\ell\}} & S_{\bT}(x_{\bC}) = \{\ell\}\\
				0 & |S_{\bT}(x_{\bC})| = 0
			\end{cases}
		\end{equation*}
		where $\overline{\br}_\ell := 1^a - \br_\ell$ is the antipode of $\br_\ell$. 
	\end{enumerate}
\end{definition}

As illustrated by \Cref{fig:uc-lb}, we associated each Talagrand term $\bT_i$ with a uniformly random bit $\bb_\ell$. If $\bb_\ell = 1$ then the action cube comprises a single union-closedness violation,\footnote{This is with the exception of $\br_\ell = 0^a$ or $1^a$; in this case $\br_\ell \cup \overline{\br}_\ell = 1^a$ and so the function on the action bits will indeed be union-closed. Note, however, that this only happens with probability $1/2^a$.} and if $\bb_\ell = 0$ then the action cube has zero satisfying assignments. This ensures that in expectation, the measure of a function drawn from $\Dno$ is indistinguishable from that of a function drawn from $\Dyes.$ 

\begin{claim} \label{claim:uc-no-far-from-uc}
	With probability at least 0.001, a function $\fno\sim\Dno := \Dno(n,\eps)$ satisfies $\dist(\fno, g) \geq \Omega(\eps)$ for every union-closed function $g\isazofunc$.
\end{claim}

\begin{proof}
Let \GoodTalagrand~be as in the proof of \Cref{lem:hehe2}, i.e. \GoodTalagrand$(T)$~is the event that an $\Omega(1)$-fraction of points $x\in\zo^c$ have $|S_{T}(x)| = 1$; following the argument of \Cref{lem:hehe2}, we have that  
\[\Prx_{\bT\sim\Tal(c,1)}\sbra{\GoodTalagrand(\bT)} > 0.01.\]
Fix an arbitrary $T$ in the support of $\Tal(c,1)$ such that $\GoodTalagrand(T)$ holds. Note that for each $x\in\zo^c$ with $S_{T}(x) = \{\ell\}$ for some $\ell\in[L]$, we have $\E[\bb_{\ell}] = 1/2$. By linearity of expectation and Markov's inequality, we thus know that with probability at least 99/100 there is an $\Omega(1)$-fraction of points such that $S_{T}(x) = \{\ell\}$ and $\bb_\ell=1$; as before, we write \ManyOne~to denote this event. Finally, we may assume that both \GoodTalagrand~and \ManyOne~ happen which is the case with probability at least 0.01. 

Next, say that a $\br(\ell)$ is \emph{bad} if it is equal to either $0^a$ or $1^a$ and \emph{good} otherwise; this happens with probability $2^{1-a}$. Yet another Markov argument gives that with probability at least $0.001$, we may assume that at most $L/2$ of the $r(\ell)$ draws good \emph{and} that both \GoodTalagrand~and \ManyOne~happen with probability $0.001$. Assuming all these three events happen, we will establish that $f_\no$ is far from being union closed.

Indeed, fix a setting of control bits $x_{C}$ that uniquely satisfy as Talagrand term $T_\ell$. As $b_{\ell} = 1$, the triple of points
	\[(x_{\bC}, r_{\ell}), (x_{\bC}, \overline{r}_{\ell}),~\text{and}~(1^a, r_{\ell})\]
	is a violation of union-closedness as long as $r(\ell)$ is good. In particular, as long as $\br_{\ell} \neq 0^a$ or $1^a$ we get a violation; note also that all these triples are disjoint. As we have to change at least one of these three points to make the function union closed, it follows that with probability at least $0.001$, the distance to being union-closed is $\Theta(1/2^a)$.  Recalling that $a = \log(1/\eps)$ completes the proof.
\end{proof}

\subsection{Indistinguishability of the Hard Distributions}
\label{subsec:uc-indistinguish}

In this section, we establish the indistinguishability of the distributions $\Dyes$ and $\Dno$ and prove \Cref{thm:two-sided-lb-UC}. Our proof will closely follow the approach used in \Cref{sec:hehe1} to prove a lower bound against intersectingness testers.

As before, we will write $Q_\calA := \{x^1, \ldots, x^q\} \sse\zo^n$ for the set of points queried by the algorithm. The argument will crucially rely on the fact that the only way for $\calA$ to distinguish $\Dyes$ and $\Dno$ is to draw two antipodal points from the same action cube, i.e. if there exist $x^i$ and $x^j$ such that $S_T(x^i_C) = S_T(x^j_C) = \{\ell\}$ for some $\ell\in[L]$ and $x^i_A$ and $x^j_A$ are antipodes; as before, we write \Bad~to denote this event. With $\view_{\calA}$ defined as in \Cref{sec:hehe1}, we have the following:

\begin{lemma} \label{lemma:uc-views}
	We have $\view_\calA(\Dyes|_{\overline{\Bad}}) = \view_\calA(\Dno|_{\overline{\Bad}}).$
\end{lemma}

\begin{proof}
	As before the distributions of the partition of $[n]$ into $\bC\sqcup\bA$ are identical for both $\Dyes$ and $\Dno$, so we may fix an arbitrary partition. As the distribution of the Talagrand DNF $\bT\sim\Tal(c,1)$ is also identical, we can fix an arbitrary $T$.
	We define 
	\[Q_{\calA}(\ell) := \cbra{x^i : S_T(x_C^i) = \{\ell\}}.\]
	Note that the points outside $\bigcup_{\ell\in[L]}Q_{\calA}(\ell)$ do not matter as the the function is identically $0$ or $1$ for both $\Dyes$ and $\Dno$. We will abuse notation and view $Q_{\calA}(\ell)$ as a subset of the action cube $\zo^{a}$ corresponding to the Talagrand term $T_{\ell}$. 
	
	We will write $\boldf_\ell$ for the function restricted to inputs in $Q_{\calA}(\ell)$, and will write $\calA(\boldf_{\ell})$ for the sequence of answers to the queries made by $\calA$ to $\boldf_{\ell}$ (i.e. the sequence of answers to queries by $\calA$ on inputs in $Q_{\calA}(\ell)$). We will write $\view_{\calA,\ell}(\Dyes)$ (respectively $\view_{\calA,\ell}(\Dno)$) to be the distribution of $\calA(\boldf_{\ell})$ for $\boldf_{\ell}\sim\Dyes$ (respectively $\boldf_{\ell}\sim\Dno$). Since $Q_{\calA}$ is partitioned as 
	\[Q_{\calA} = \bigsqcup_{\ell\in[L]} Q_{\calA}(\ell), \]
	note that in order to show that $\view_{\calA}(\Dyes|_{\overline{\Bad}}) = \view_{\calA}(\Dno|_{\overline{\Bad}})$, it suffices to show that $\view_{\calA,\ell}(\Dyes|_{\overline{\Bad}}) = \view_{\calA,\ell}(\Dno|_{\overline{\Bad}})$; this is what we will establish below.
	
	Fixing an action cube $\zo^a$ (which is indexed by $\ell\in[L]$), note that the actions cubes in the yes- and no-distributions can be equivalently described as follows: 
	\begin{enumerate}
		\item Draw a uniformly random pair of points $(\by,\overline{\by})$ from the $2^{a-1}$ pairs $(x, \overline{x})$ for $x\in\zo^a$, and draw a uniformly random bit $\bb_{\ell}$. 
		\item We consider the ``yes'' and ``no'' cases separately:
		\begin{enumerate}
				\item In the ``yes'' case, if $\bb_{\ell} = 1$, then set $\bs_{\ell} = \by$; otherwise set $\bs_{\ell} = \overline{\by}$. 
				\item In the ``no'' case, set $(\br_{\ell},\overline{\br}_{\ell}) = (\by,\overline{\by})$; and if $\bb_{\ell} = 0$, then the function $\boldf|_{\ell}$ is defined to be identically zero on the action cube (cf.~\Cref{def:uc-no} and \Cref{fig:uc-lb}). 
		\end{enumerate}
	\end{enumerate}
Note that conditioned on $\Bad$ not happening, we have that none of the query points in $Q_{\calA}(\ell)$ are antipodes of each other. We now split into two cases depending on whether either $\by$ or $\overline{\by}$ is in the query set $Q_{\calA}(\ell)$:
\begin{enumerate}
	\item If $\by,\overline{\by}\notin Q_{\calA}(\ell)$, then note that $\view_{\calA,\ell}(\Dyes|_{\overline{\Bad}}) = \view_{\calA,\ell}(\Dno|_{\overline{\Bad}})$ since $\boldf_\ell$ is identically $0$ on $Q_{\calA}(\ell)$ in both the ``yes'' and the ``no'' cases. 
	\item Otherwise, since we conditioned on $\overline{\Bad}$, only one of $\by,\overline{\by}$ can be in $Q_{\calA}(\ell)$; without loss of generality, suppose that it is $\by$. In both the ``yes'' and the ``no'' cases, $\by$ is a $1$-input if and only if $\bb_{\ell}=1$, and the function is identically $0$ on all other points. (Recall that we view points of $Q_{\calA}(\ell)$ as a subset of the action cube $\zo^a$ corresponding to the Talagrand DNF term $T_{\ell}$.)
\end{enumerate}
	It follows that $\view_{\calA,\ell}(\Dyes|_{\overline{\Bad}}) = \view_{\calA,\ell}(\Dno|_{\overline{\Bad}})$, and since $Q_{\calA}$ is partitioned by the indices $\ell\in[L]$, we have 
	\[\view_{\calA}(\Dyes|_{\overline{\Bad}}) = \view_{\calA}(\Dno|_{\overline{\Bad}}),\]
	completing the proof.
\end{proof}

Next, we will show that \Bad~happens with $o_n(1)$ probability:

\begin{lemma} \label{lemma:uc-bad}
	For any set of points $Q_{\calA} = \{x^1,\ldots,x^q\} \sse\zo^n$ where $q := n^{{0.001\log(1/\eps)}}$, we have 
	\[\Prx[\Bad] = o_n(1).\]
\end{lemma}

\begin{proof}
	For $x,y\in\zo^n$, let $\Bad_{xy}$ be the event that $S_{\bT}(x_{\bC}) = S_{\bT}(y_{\bC}) = \{\ell\}$ for some $\ell\in[L]$ and $x_{\bA} = \overline{y}_{\bA}$. We will upper bound the probability of $\Bad_{xy}$ in what follows.

	Let $J \sse[n]$ be the coordinates in which $x$ and $y$ differ, i.e. $J := \{i \in [n] : x_i \neq y_i\}$. Define the event $\diamond$ as:
	\begin{equation} \label{eq:diamond}
		\bA \sse J. \tag{$\diamond$}
	\end{equation}
	We also define the event $\star$ as before as 
	\begin{equation}
		\text{There exists an } \ell \in [L]\text{ such that }S_{\bT}(x)=S_{\bT}(y)=\{\ell\}. \tag{$\star$}
	\end{equation}
	By definition of $\Bad_{xy}$, we have that 
	$\Pr[\Bad_{xy}]\leq\min\cbra{\Pr[\star],\Pr[\diamond]}$. In the rest of the proof, we will establish that 
	\begin{equation}
		\min\cbra{\Pr[\star],\Pr[\diamond]} \leq \Theta\pbra{\frac{1}{n}}^{0.01 a},\label{eq:uc-indistinguish-goal}
	\end{equation}
	from which the lemma follows immediately by taking a union bound over all $q^2$ pairs $(x,y) \in Q_\calA\times Q_\calA $.
	
	Note that 
	\begin{align*}
		\Prx[\diamond] = \Prx\sbra{\bA\sse J} \leq \pbra{\frac{e|J|}{n}}^a
	\end{align*}
	via standard bounds on binomial coefficients. On the other hand, proceeding as in the proof of~\Cref{lemma: Bad is unlikely}, we have 
	\begin{align*}
		\Prx[\star] &\leq \max_{\ell\in[L]}\Prx\sbra{S_{\bT}(x) = S_{\bT}(y) \mid S_{\bT}(y) = \{\ell\}} \\
		&\leq \pbra{1 - \frac{1}{\sqrt{c}}}^{|J|} \leq \exp\pbra{\frac{-|J|}{\sqrt{c}}}
	\end{align*}
	where the final line follows from the definition of $\Tal(c,1)$.
	In particular, note that if $|J| \leq n^{0.5}$, then
	\[\Pr[\star] \leq \pbra{\frac{e}{n^{0.5}}}^a,\]
	and if $|J| > n^{0.5}$ then we have 
	\[\Pr[\diamond] \leq \exp\pbra{\frac{-n^{0.5}}{\sqrt{n - \log(1/\epsilon)}}} \ll \pbra{\frac{1}n}^{\Theta(a)}\]
	where the final inequality uses the fact that $\eps \geq \Theta\pbra{\frac{1}{2^{n^{0.49}}}}$. Putting everything together establishes \Cref{eq:uc-indistinguish-goal} which in turn completes the proof.
\end{proof}

\Cref{thm:two-sided-lb-UC} follows from \Cref{lemma:uc-bad,lemma:uc-views} \emph{mutatis mutandis} as \Cref{thm:main-lb-intersecting} follows from \Cref{lemma: Bad is the only different place between Dy and Dn,lemma: Bad is unlikely}.


\section{Upper Bounds}
\label{sec:ubs}

We describe our algorithms for testing union-closedness and intersectingness in \Cref{subsec:alg-uc,subsec:alg-int} respectively.

\subsection{Upper Bound for Testing Union-Closedness} 
\label{subsec:alg-uc}

We will first give some useful definitions in the context of union-closed functions. Subsequently, we will give a tester for union-closedness. 

\begin{definition}~\label{def:UC-violating}
For a function $f: \{0,1\}^n \rightarrow \{0,1\}$, a sequence  $(x_1, \ldots, x_k, x_1 \cup \ldots \cup x_k)$ (where each $x_i \in \{0,1\}^n$) is said to be a \emph{union-closed violating tuple} or \emph{\textsf{UC}-violating tuple} for short if the following condition holds:
\[\text{For each}~1 \le j \le k, f(x_j)=1 ~\text{and}~f(x_1 \cup \ldots \cup x_k)=0.\]
Furthermore, we say that two \textsf{UC}-violating tuples $(x_1, \ldots, x_k, x_1 \cup \ldots \cup x_k)$ and $(y_1, \ldots, y_k, y_1 \cup \ldots \cup y_k)$ are {\em end-distinct} if $x_1 \cup \ldots \cup x_k \not = y_1 \cup \ldots \cup y_k$.  

Finally, we say that a \textsf{UC}-violating tuple $(x_1, \ldots, x_k, x_1 \cup \ldots \cup x_k)$ is \emph{minimal} if for any $1 \le j \le k$, it holds that $\cup_{i \not =j} x_i \subsetneq (x_1 \cup \ldots \cup x_k)$. 
\end{definition} 

Given $f: \zo^n \to \zo$, we write $\duc(f)$ to denote the distance from $f$ to the property which consists of all union-closed functions, i.e.~
\[
\duc(f) = \min_{\substack{g\isazofunc\\g\text{~is~union-closed}}} \dist(f,g).
\]
The following lemma shows that any function which is far from union-closed must have many end-distinct violating tuples:
\begin{lemma}~\label{lem:distance-uc}
Suppose $f: \{0,1\}^n \rightarrow \{0,1\}$ is such that $\duc(f) \geq \epsilon$. Then, $f$ has at least $\epsilon \cdot 2^n$ many  end-distinct \textsf{UC}-violating
tuples.
\end{lemma}

\begin{proof}
Let ${\cal C}$ be a maximal collection of end-distinct $\mathsf{UC}$-violating tuples for $f$; {note that every string $z \in \{0,1\}^n$ which is the last coordinate of some $\mathsf{UC}$-violating tuple occurs as the last coordinate of some $\mathsf{UC}$-violating tuple in ${\cal C}$}. Let ${\cal B}$ be the set of all the points in $\{0,1\}^n$ which appear as the last coordinate of a tuple in ${\cal C}$. Define ${\cal A} =\{0,1\}^n \setminus {\cal B}$.

We now define $\tilde{f}: \{0,1\}^n \rightarrow \{0,1\}$ as follows: 
\begin{enumerate}
\item $\tilde{f}(x)=1$ if there are $x_1,\ldots, x_k \in {\cal A}$ such that $f(x_1) = \ldots = f(x_k)=1$. 
\item $\tilde{f}(x)=0$ otherwise. 
\end{enumerate}
We now make two claims about our construction $\tilde{f}$. 
\begin{claim}~\label{clm:tildef-1}
For any point $z \in {\cal A}$, $f(z) = \tilde{f}(z)$. 
\end{claim} 
\begin{proof}
Note that for points $z \in {\cal A}$ where $f(z)=1$, $\tilde{f}(z)=1$ by construction. Further, $\tilde{f}(z)=1$ iff there are points $x_1, \ldots, x_k \in {\cal A}$ such that $f(x_1) = \ldots = f(x_k)=1$. Since ${\cal A}$ does not contain the last coordinate of any \textsf{UC}-violating tuple, it follows that if $f(z) =0$, then $\tilde{f}(z) =0$. 
\end{proof}
\begin{claim}~\label{clm:tildef-2}
The function $\tilde{f}$ is union-closed. 
\end{claim}
\begin{proof}
Towards contradiction, let $(y_1, \ldots, y_k, z)$ be a \textsf{UC}-violating tuple for $\tilde{f}$. This means that for each $1 \le j \le k$, $\tilde{f}(y_j)=1$. By construction of $\tilde{f}$, this means that for each $j$, there are points $y_{j,1}, \ldots, y_{j, \ell_j}$ such that (i) $f(y_{j_1}) = \ldots = f(y_{j, \ell_j}) =1$;  (ii)  $y_j = \cup_{1 \le i \leq \ell_j} y_{j,i}$. 

This means that $z = \cup_{j=1}^k y_j =  \cup_{j=1}^{k} \cup_{i=1}^{\ell_j} y_{j,i}$ 
and for all $1 \le j \le k$ and $1 \le i \le \ell_j$, $f(y_{j,i})=1$. By construction, this implies that $\tilde{f}(y)=1$, contradicting the assumption that $(y_1, \ldots, y_k, z)$ is a \textsf{UC}-violating tuple for $\tilde{f}$.
\end{proof}

Consequently, there is a union-closed function $\tilde{f}$ all of whose disagreements with $f$ are on element s of ${\cal B}$. Since 
$\duc(f) \geq \epsilon$, this implies that $|{\cal B}| \geq \epsilon \cdot 2^n$. This completes the proof of \Cref{lem:distance-uc}.  
\end{proof}

We next record the following simple claim, the proof of which is left to the reader. 
\begin{claim}~\label{claim:zeroing}
Let $f: \zo^n \rightarrow \zo$. For $0<\epsilon<1$, let $T:=\sqrt{n \cdot 2 \ln (4/\epsilon)}$ and define the Boolean function $f_{\mathsf{trunc}}: \zo^n \rightarrow \zo$ as 
\[
f_{\mathsf{trunc}}(x) = \begin{cases} 0  \ \ &\textrm{if} \  | x | \le n/2 - T \\
1 \ \ &\textrm{if} \  | x | \ge n/2 + T \\
f(x) \ \ &\textrm{otherwise}
\end{cases}.
\]
Then, the following hold:
\begin{enumerate}
	\item If $f$ is union-closed, then so is $f_{\mathsf{trunc}}$;
	\item If $\duc(f)\ge\epsilon$, then $\duc(f_{\mathsf{trunc}}) \ge\epsilon/2$; and  
	\item If $(y_1, \ldots, y_T,y_{T+1})$ is a $\mathsf{UC}$-violating tuple for $f_{\mathsf{trunc}}$, then for each $i \in [T+1]$, the point $y_i$ lies in ${[n]\choose j}$
for some $j \in [n/2-T, n/2+T]$.
\end{enumerate}
\end{claim}

\begin{algorithm}[t]
\caption{Algorithm to Test Union-Closedness}
\label{alg:union-closed}
\vspace{0.5em}
\textbf{Input:} Query access to $f: \{0,1\}^n \rightarrow \{0,1\}$, parameter $0 < \eps < 1$\\[0.5em]
\textbf{Output:} ``Union-closed'' or  ``$\eps$-far from union-closed''

\ 

\textsf{Union-Closed-Tester}$(f)$:
\begin{enumerate}
	\item Repeat the following $M:= 100/\epsilon$ times:
		\begin{enumerate}
			\item Set $I:= [n/2 - \sqrt{n \cdot 2 \ln (4/\epsilon)}, n/2 + \sqrt{n  \cdot 2 \ln (4/\epsilon)}].$ 
			\item Sample $\bx$ uniformly from the set ${[n]\choose I}$.
			\item Let $\bx_{\downarrow}:= \{y \le \bx: y \in {[n]\choose I}\}$. Query $f(y)$ for every $y$ in $\bx_{\downarrow}$.
			\item Check if there are $y_1, \ldots, y_T \in \bx_{\downarrow}$ such that $(y_1, \ldots, y_T, \bx)$ is a \textsf{UC}-violating tuple. 
			\item If yes, halt and output ``$\eps$-far from union-closed".
		\end{enumerate}
	\item Output ``union-closed''.
\end{enumerate}
\end{algorithm}

\Cref{alg:union-closed} gives our algorithm for testing union-closedness, and the following theorem establishes its performance and correctness:

\begin{theorem}~\label{thm:Union-closed-algorithm}
Given oracle access to $f: \zo^n \rightarrow \zo$ and error parameter $0< \epsilon\leq 1/2$, the algorithm \textsf{Union-Closed-Tester} (cf.~\Cref{alg:union-closed}) makes $n^{O(\sqrt{n \log (1/\epsilon)})}/\eps$ queries to $f$ and has the following performance guarantee: 
\begin{enumerate}
\item If the function $f$ is union-closed, then the algorithm returns ``union-closed" with probability $1$. 
\item If $\duc(f) > \epsilon$, then the algorithm returns ``$\eps$-far from union-closed" with probability at least $9/10$. 
\end{enumerate}
\end{theorem}
\begin{proof}
Note that for any $x \in {[n]\choose I}$, $|x_{\downarrow}| = n^{O(\sqrt{n \log (1/\epsilon)})}$, from which the query complexity of \textsf{Union-Closed-Tester} is immediate. 

Note that the first item is also immediate because if $f$ is union-closed, then in Step~1(d), the algorithm must fail to find a $\mathsf{UC}$-violating tuple and hence it will always output ``union-closed" in Step~2. 

Finally, to prove the second item, note that if $\duc(f) > \epsilon$,  then by \Cref{claim:zeroing} we have $\duc(f_{\mathsf{trunc}}) > \epsilon/2$. Let us call a point $x \in \zo^n$ a ``witness" if there exists $\{y_i\}_{1 \le i \le T}$ such that $(y_1, \ldots, y_T, x)$ is a $\mathsf{UC}$-violating tuple for $f_{\mathsf{trunc}}$. By Claim~\ref{claim:zeroing}, note that any such $x$ and the corresponding $y_1, \ldots, y_T$ must lie in $ {[n]\choose I}$.

Applying~\Cref{lem:distance-uc}, the probability that $\bx$ sampled in Step~1(b) is a witness is at least $\epsilon/2$. Further, if $x$ is a witness, then in Step~1(d), the algorithm will find a $\mathsf{UC}$-violating tuple $(y_1, \ldots, y_T, x)$. Thus each iteration of lines (a)--(e) in the algorithm will output ``far-from union-closed" with probability at least $\epsilon/2$. This means that repeating the procedure $M=100/\epsilon$ times, we get the correct answer with probability at least $9/10$. 
\end{proof}

\subsubsection{Extension to a Triple Tester} \label{sec:extension}

The algorithm in \Cref{thm:Union-closed-algorithm} (cf.~\Cref{alg:union-closed})
checks for \textsf{UC}-violating tuples. As stated in the introduction, our algorithm can be modified into a so-called ``triple tester", i.e. an algorithm checks for \textsf{UC}-violating triples in the function $f$. In other words, it looks for triples $(y_1, y_2,x)$ such that $f(x)=0$, $f(y_1) = f(y_2) =1$ and $y_1 \cup y_2 =x$. In this section, we briefly sketch how this modification can be done.

Given any function $f: \zo^n \rightarrow \zo$ and a corresponding \textsf{UC}-violating tuple $\tau$ defined as $\tau=(x_1, \ldots, x_k, x_1 \cup \ldots \cup x_k)$, we define its augmentation $\widetilde{\tau}$ as
\[
\widetilde{\tau} := (x_1, \ldots, x_k, x_1 \cup x_2, x_1 \cup x_2 \cup x_3, \ldots, x_1 \cup \ldots  \cup x_k). 
\]
Note that whenever $\tau$ is  \textsf{UC}-violating tuple, there is some choice of $1 \le j <k$, such that $(x_1\cup \ldots \cup x_j, x_{j+1}, x_1\cup \ldots \cup x_{j+1})$ is a \textsf{UC}-violating triple. 
 We now have the following easy but crucial observation:
 
\begin{fact}~\label{fact:augmented}
If $f$ has $M$ \textsf{UC}-violating tuples $\tau_1, \ldots, \tau_M$ such that the augmentations $\widetilde{\tau_1}, \ldots, \widetilde{\tau_M}$ are pairwise-disjoint, then $f$ has $M$ \textsf{UC}-violating triples $\sigma_1, \ldots, \sigma_M$. Further, if all the coordinates of the tuples 
$\tau_1, \ldots, \tau_M$ lie in $\binom{[n]}{I}$ for an interval $I$, then so do the coordinates of $\sigma_1, \ldots, \sigma_M$. 
\end{fact}
We next have the following analogue of \Cref{lem:distance-uc}. 
\begin{lemma}~\label{lem:distance-uc1}
Suppose $f: \{0,1\}^n \rightarrow \{0,1\}$ is such that $\duc(f) > \epsilon$. Then, $f$ has at least $(\epsilon/n)\cdot 2^n$ many  disjoint \textsf{UC}-violating
tuples. 
\end{lemma}
\begin{proof}
Let ${\cal C}$ be a maximal collection of disjoint minimal \textsf{UC}-violating
tuples for $f$. Let ${\cal B}$ be the set of all points in $\zo^n$ which appear in any of tuples in ${\cal B}$. Note that any violating tuple for $f$ must have at least one of its coordinates in ${\cal B}$. 

With the definition of ${\cal B}$, if we define ${\cal A}$ and $\tilde{f}$ exactly as in the proof of \Cref{lem:distance-uc}, then the proofs of \Cref{clm:tildef-1} and \Cref{clm:tildef-2} will go through and we will obtain that $\tilde{f}$ is union-closed. Since $\duc(f) > \epsilon$, it follows that $|{\cal B}| > \epsilon \cdot 2^n$ and thus, $|{\cal C}| > \epsilon/n \cdot 2^n$. This finishes the proof. 
\end{proof}

Using \Cref{fact:augmented} and \Cref{lem:distance-uc1},
we get the following claim. 

\begin{claim}~\label{clm:disjoint-triples}
If $\duc(f) \ge \epsilon$, then 
$f$ has at least $\epsilon/ n \cdot 2^n$ disjoint violating triples. 
\end{claim}
To define the triple tester, we now consider the following ``basic subroutine":
\begin{enumerate}
	\item Set 
$I:= [n/2 - \sqrt{n \cdot 2 \ln (4n/\epsilon)}, n/2 + \sqrt{n  \cdot 2 \ln (4n/\epsilon)}].$ (Note that the ``width" of $I$ is different from the setting in \Cref{alg:union-closed} by a factor of $\sqrt{\ln n}$.)
	\item Sample $\bx$ 
 uniformly from the set $\binom{[n]}{I}$ and sample $\by_1$ and $\by_2$ uniformly and independently at random from $\bx_{\downarrow}$.
 	\item If  $(\by_1, \by_2, \bx)$ form a $\mathsf{UC}$-violating triple, then output ``$\epsilon$-far from union-closed".
\end{enumerate}

Observe that if $f$ is union-closed, then the basic subroutine will not output ``$\epsilon$-far from union-closed". On the other hand, if $f$ is $\epsilon$-far from union-closed, then using \Cref{clm:disjoint-triples} and  the same reasoning as \Cref{thm:Union-closed-algorithm}, it follows that the basic subroutine will find a $\mathsf{UC}$-violating triple with probability at least $\tau:=\epsilon \cdot 2^{-\Theta(\sqrt{n \log(n/\epsilon)} \log n)}
$. 

To amplify the soundness of our algorithm, we simply repeat the basic subroutine $100/\tau$ times. If any of this invocations returns ``$\epsilon$-far from union-closed", then we output ``$\epsilon$-far from Union-closed". Otherwise, we return ``union-closed". This preserves the completeness of the procedure but we also have the guarantee that if the function $f$ is such that $\duc(f) \ge \epsilon$, the algorithm returns ``$\epsilon$-far from union-closed" with probability at least $9/10$. The query complexity of the procedure is now $\Theta(1/\tau) = 2^{\Theta(\sqrt{n \log(n/\epsilon)} \log n)} /\epsilon$.

\subsection{Upper Bound for Testing Intersectingness}
\label{subsec:alg-int}

We start with some useful definitions for intersecting sets, and then give our tester for intersecting functions. 

\begin{definition}\label{def:intersecting-violation}
For a function $f: \zo^n \rightarrow \zo$, a pair $(x,y)$ is said to be an \emph{intersection violation} (which we abbreviate ``$\mathsf{I}$-violation'')  if $f(x) = f(y)=1$ and $x \cap y = \emptyset$. 
\end{definition}

Given $f: \zo^n \to \zo$, we write $\dint(f)$ to denote the distance from $f$ to the property which consists of all intersecting functions, i.e.~
\[
\dint(f) = \min_{\substack{g: \zo^n \to \zo\\g\text{~is~intersecting}}} \dist(f,g).
\]
The following lemma is similar in spirit to \Cref{lem:distance-uc} and shows that any function which is far from intersecting must have many disjoint violating pairs. In fact, it gives a characterization of the distance from a function $f$ to the property of being intersecting up to a factor of two (we note that both inequalities below can be sharp):

\begin{lemma} \label{lem:dist-to-intersecting}
	Given a boolean function $f:\{0,1\}^n\rightarrow\{0,1\}$, we have 
	$$
	|M|/2^n \leq \dint(f) \leq 2|M|/2^n,
	$$
	where $M$ is a maximum sized set of disjoint \textsf{I}-violating pairs of $f$.
\end{lemma}
\begin{proof}
	Let $M=\{(x^1,y^1),\cdots,(x^{|M|},y^{|M|})\}$ be a maximum sized set of disjoint \textsf{I}-violating pairs of $f$.
	
	Observe that any intersecting function $g$ must satisfy $g(x^i)=0$ or $g(y^i)=0$ for all $1\leq i\leq |M|$. This means the distance between $f$ and $g$ is at least $M/2^n$, so $\dint(f)\geq |M|/2^n$.
	
	For the other direction, we show that the function $g$ defined below satisfies that $g$ is intersecting and $\dist(f,g)=2|M|/2^n$ (so $\dint(f) \leq 2|M|/2^n$):
	
\begin{equation*}
g(x) =
	\begin{cases}
       \ 0 & x=x^i \text{ or }x=y^i\text{ for some }1\leq i\leq |M|; \\[0.5ex]
       \ f(x) & \text{otherwise}.
    \end{cases}
\end{equation*}
It is clear that $\dist(f,g)=2|M|/2^n$. To see that $g$ is intersecting, suppose that there are $x,y$ such that $g(x)=g(y)=1$ and $x\cap y=\varnothing$. Then it must be the case that 
$x,y\not\in\{x^1,\cdots,x^{|M|},y^1,\cdots,y^{|M|}\}$, so $f(x)=g(x)=1$ and $f(y)=g(y)=1$, which violates the assumption that $M$ is maximal.
\end{proof}

Similar to \Cref{claim:zeroing}, the following claim says that for any $f$, we can truncate $f$ to its middle layers and distance to intersectingness is essentially preserved. It admits a simple proof which we leave to the reader. 
\begin{claim}~\label{claim:zeroing-inter}
Let $f: \zo^n \rightarrow \zo$. For $0<\epsilon<1$, define $T:=\sqrt{n \cdot 2\ln (4/\epsilon)}$ and define $f_{\mathsf{trunc}}: \zo^n \rightarrow \zo$ as 
\[
f_{\mathsf{trunc}}(x) = \begin{cases} 0  \ \ &\textrm{if} \  | x | \le n/2 - T; \\
0 \ \ &\textrm{if} \  | x | \ge n/2 + T; \\
f(x) \ \ &\textrm{otherwise}
\end{cases}
\]
Then, the following hold: (i) If $f$ is intersecting, then so is $f_{\mathsf{trunc}}$. (ii)  If $\dint(f)\ge\epsilon$, then $\dint(f_{\mathsf{trunc}}) \ge\epsilon/2$. Finally, (iii) If $(x,y)$ is a $\mathsf{I}$-violating pair for $f_{\mathsf{trunc}}$, then each of $x$ and $y$ lies in ${[n]\choose j}$
for some $j \in [n/2-T, n/2+T]$. 
\end{claim}

\begin{algorithm}[t]
\caption{Algorithm to Test Intersectingness}
\label{alg:intersecting}
\vspace{0.5em}
\textbf{Input:} Query access to $f: \{0,1\}^n \rightarrow \{0,1\}$,
parameter $0 < \eps < 1$\\[0.5em]
\textbf{Output:} ``Intersecting'' or  ``$\eps$-far from intersecting''

\ 

\textsf{Intersecting-Tester}$(f)$:
\begin{enumerate}
	\item Repeat the following $M:= 100/\epsilon$ times:
	\begin{enumerate}
		\item Set $I:= [n/2 - \sqrt{n \cdot 2\ln (4/\epsilon)}, n/2 + \sqrt{n  \cdot 2\ln (4/\epsilon)}].$ 
	    \item Sample $\bx$ uniformly from the set ${[n]\choose I}$. 
	    \item Let $\overline{\bx}$ be obtained by flipping every bit of $\bx$. 
	    \item Let $\overline{\bx_{\downarrow}}:= \{y \le \overline{\bx}: y \in {[n]\choose I}\}$. Query $f(y)$ for every $y$ in $\overline{\bx_{\downarrow}}$. 
	    \item 
	    Check if there is a $y \in \overline{\bx_{\downarrow}}$ such that $f(y) = f(\bx)=1$.  
	    \item  If yes, output ``$\eps$-far from intersecting". 
	\end{enumerate}
	\item Output ``intersecting''.
\end{enumerate}
\end{algorithm}

\begin{theorem}~\label{thm:Intersecting-algorithm}
Given oracle access to $f: \zo^n \rightarrow \zo$ and error parameter $0<\eps \leq \red{1/2}$, the algorithm \textsf{Intersecting-Tester} makes $n^{O(\sqrt{n \log (1/\epsilon)})}/\eps$ queries to $f$ and has the following performance guarantee: 
\begin{enumerate}
\item If the function $f$ is intersecting, then the algorithm returns ``intersecting" with probability $1$. 
\item If $\dint(f) > \epsilon$, then the algorithm returns ``$\eps$-far from intersecting" with probability at least $9/10$. 
\end{enumerate}
\end{theorem}
\begin{proof}
Note that for any $x \in {[n]\choose I}$, $|\overline{x_{\downarrow}}| = n^{O(\sqrt{n \log (1/\epsilon)})}$ and so the claimed query complexity of the algorithm follows immediately. 

The first item is also immediate, because if $f$ is intersecting, then in Step~1(d), the algorithm will fail to find an $\mathsf{I}$-violating pair. Thus, it will always output ``intersecting" in Step~2. 

Finally, to prove the second item, note that if $\dint(f) > \epsilon$,  then $\dint(f_{\mathsf{trunc}}) > \epsilon/2$ (where $f_{\mathsf{trunc}}$ is defined in  \Cref{claim:zeroing-inter}). Now, call a point $x \in \zo^n$ a ``witness" if there exists $y \in \zo^n$ such that $(y, x)$ is an $\mathsf{I}$-violating pair for $f_{\mathsf{trunc}}$. By \Cref{claim:zeroing}, note that any such $x$ and the corresponding $y$ must lie in $ {[n]\choose I}$.

Applying~\Cref{lem:dist-to-intersecting}, the probability that $\bx$ sampled in Step~3 is a witness is at least $\epsilon/2$. Further, if $x$ is a witness, then in Step~5, the algorithm will find a $\mathsf{I}$-violating pair $(y, x)$. Thus each iteration of lines (a)--(f) in the algorithm will output ``$\eps$-far from intersecting" with probability at least $\epsilon/2$. This means that repeating the procedure $M=100/\epsilon$ times, we get the correct answer with probability at least $9/10$. 
\end{proof}

We remark that similar to \Cref{sec:extension}, a straightforward modification of the above algorithm yields a ``pair tester'' for the property of  intersectingness with a comparable query complexity. 


\section{Discussion:  Towards Sharper Bounds for Testing Union-Closed Families?}
\label{sec:discussion}

An appealing goal for future work is to try to narrow the gap between our $n^{\Omega(\log 1/\eps)}$-query lower bound and our $\poly(n^{\sqrt{n \log 1/\eps}},1/\eps)$-query upper bound for non-adaptive testing of union-closed families.  Towards this end, it would be helpful to have a better understanding of the types of violating triples which must be present in functions which are far from union-closed.

In more detail, given a violating triple $(x,y,x \cup y)$ for union-closedness, let us say that the \emph{locality} of the triple is 
\[
\pbra{|x \cup y|_1 - |x|_1} + \pbra{|x \cup y|_1 - |y|_1},
\]
i.e.~the size of the symmetric difference of the two sets $x$ and $y$.
The arguments of \Cref{sec:extension} show that if $\duc(f) \geq \epsilon$, then $f$ must have at least $\Omega(\eps/n) \cdot 2^n$ many triples, each of which has locality $O(\sqrt{n \log(1/\eps)})$, and this simple structural result is at the heart of the triple tester sketched in that section.  

We (somewhat loosely) conjecture that if $f: \zo^n \to \zo$ has $\duc(f) \geq \eps$, then in fact $f$ must have ``many'' violating triples of locality $\ell$ for some locality parameter $\ell \ll \sqrt{n \log(1/\eps)}$.  If such a structural result were true with suitable parameters, then this would directly yield more query-efficient union-closed triple-based testing algorithms than the one given by \Cref{thm:main-alg}.  On the other hand, if there exist functions $f: \zo^a \to \zo$ that are $\eps$-far from union-closed but whose only violating triples have large locality, then such functions might be useful for improving the lower bound construction of \Cref{thm:one-sided-lb-intersecting}.  (Recall that at the core of the proof of \Cref{thm:one-sided-lb-intersecting} is the simple function $g_r: \zo^a \to \zo$ which has two satisfying assignments, which are antipodal strings $r,\overline{r} \in \zo^a$, where $r,\overline{r} \neq 0^a$ and $a=\log_2(1/\eps)$.  This function is $\eps$-far from being union-closed but its only violating triple has locality $a=\log_2(1/\eps).$)

As a concrete first goal for future work in this direction, we pose the following question:  

\begin{question}
Let $f: \zo^n \to \zo$ satisfy $\duc(f) \geq \epsilon$. What is the smallest value $\ell_f$ such that $f$ must contain some violating triple of locality at most $\ell_f$?
\end{question}

\section*{Acknowledgements}
X.C. is supported by NSF grants IIS-1838154, CCF-2106429, and CCF-2107187. A.D. is supported
by NSF grants CCF-1910534 and CCF-2045128. Y.L. is supported by NSF grants IIS-1838154,
CCF-2106429 and CCF-2107187. S.N. is supported by NSF grants IIS-1838154, CCF-2106429,
CCF-2211238, CCF-1763970, and CCF-2107187. R.A.S. is supported by NSF grants IIS-1838154,
CCF-2106429, and CCF-2211238. This work was partially completed while some of the authors
were visiting the Simons Institute for the Theory of Computing at UC Berkeley.

\begin{flushleft}
\bibliographystyle{alpha}
\bibliography{allrefs}
\end{flushleft}

\appendix

\end{document}